\documentclass[runningheads,envcountsame]{llncs}

\usepackage[utf8]{inputenc}
\usepackage[english]{babel}
\usepackage[babel]{csquotes}
\usepackage[T1]{fontenc}

\usepackage{graphicx}
\usepackage{hyperref}

\usepackage{amsmath}
\usepackage{amssymb}
\usepackage[nameinlink]{cleveref}

\usepackage{subcaption}

\usepackage[appendix=append,bibliography=common]{apxproof}

\usepackage{hyperref}
\AtBeginDocument{\hypersetup{breaklinks=true,pdfborder={1 1 1},linkbordercolor=[rgb]{1,0.8,0.8},citebordercolor=[rgb]{0.6,1,0.6}}}

\usepackage{amssymb,amsmath,latexsym,stmaryrd}
\usepackage{stmaryrd}
\usepackage{mathtools}
\usepackage{microtype}
\usepackage{csquotes}
\usepackage{enumerate}
\usepackage{ifthen}
\usepackage{xifthen}
\usepackage{graphicx,epsfig,color}
\usepackage{tikz}
\usepackage{adjustbox}
\usepackage{xfrac}
\usepackage{enumitem}

\makeatletter
\def\Gin@getbase#1{%
  \edef\Gin@base{\filename@area\filename@base}%
  \edef\Gin@ext{#1}%
}
\makeatother

\usepackage{amsthm}
\usepackage{cleveref}
\usepackage{soul}

\newcommand{\renewtheorem}[1]{%
  \expandafter\let\csname #1\endcsname\relax%
  \expandafter\let\csname c@#1\endcsname\relax%
  \expandafter\let\csname end#1\endcsname\relax%
  \newtheorem{#1}%
}
\theoremstyle{plain}
  \newtheoremrep{theorem}{Theorem}
	\newtheoremrep{lemma}[theorem]{Lemma}
	\newtheoremrep{corollary}[theorem]{Corollary}
	\newtheoremrep{proposition}[theorem]{Proposition}
	\newtheoremrep{exercise}[theorem]{Exercise}
	\newtheoremrep{definition}[theorem]{Definition}
	\theoremstyle{definition}
	\newtheoremrep{example}[theorem]{Example}
	\theoremstyle{remark}
	\newtheoremrep{note}[theorem]{Note}
	\newtheoremrep{remark}[theorem]{Remark}
	\theoremstyle{claimstyle}
	\newtheoremrep{claim}[theorem]{Claim}
\theoremstyle{plain}

\usetikzlibrary{matrix,shapes.geometric,backgrounds,fit}
\usetikzlibrary{positioning,shapes,arrows.meta,calc,bbox}

\tikzset{%
  roundbox/.style={rectangle, rounded corners=4pt},
  >=To,
  gedge/.style={->,>=latex},
  condtri/.style={draw,dart,dart tail angle=135,dart tip angle=50,inner xsep=0.1em, inner ysep=0.2em,anchor=tip},
  fancydotted/.style={dash pattern=on 1.25pt off 1.75pt},
  gnfont/.style={font=\scriptsize},
  gn/.style={gnfont,circle,draw,inner sep=0.2pt, minimum size=8pt},
  exaloop above/.style={out=110,in=70,loop,looseness=9,/pgf/bezier bounding box=true},
  exaloop below/.style={out=-70,in=-110,loop,looseness=9,/pgf/bezier bounding box=true},
  exaloop right/.style={out=25,in=-15,loop,looseness=9,/pgf/bezier bounding box=true},
  exaloop left/.style={out=-155,in=165,loop,looseness=9,/pgf/bezier bounding box=true},
}
\definecolor{graphmorphismgreen}{rgb}{0.4,0.84,0.3}
\definecolor{graphmorphismfailred}{rgb}{0.9,0.2,0.4}
\tikzset{
  cospanintcommon/.style={draw=black!38,line width=0.5pt},
  cospanint/.style={->,cospanintcommon},
  cospanintmono/.style={cospanintcommon,>->},
  cospanarr/.style={->,line width=0.5pt},
  graphmor/.style={cospanarr},
  graphmorcol/.style={->,line width=0.5pt,color=graphmorphismgreen},
  graphmorcolfail/.style={->,dash pattern=on 1.25pt off 1.75pt,line width=0.5pt,color=graphmorphismfailred},
}

\pgfdeclarelayer{bg}    
\pgfdeclarelayer{background}    
\pgfsetlayers{bg,background,main}  

  \newcommand{\dotEx}{\ .\ \exists\ }
  
  \newcommand{\dotFalse}{\ . \condfalse}
  \newcommand{\dotTrue}{\ . \condtrue}
  \def\saveBboxNorthAs#1{\coordinate (#1) at ($(current bounding box.north)+(0pt,0.8pt)$);}
  \def\saveBboxSouthAs#1{\coordinate (#1) at ($(current bounding box.south)+(0pt,-0.8pt)$);}
  \def\faintborder{%
    \begin{pgfonlayer}{bg}
      \draw[black!40, fill=black!5, rounded corners=2pt, line width=0.2pt, overlay]
      ($(current bounding box.south west)+(-1pt,-1pt)$) rectangle
      ($(current bounding box.north east)+(1pt,0.8pt)$);
    \end{pgfonlayer}
  }
  
  \def\fcGraph#1#2{
    \tikz[baseline=(#1.base),x=0.6cm]{#2 \faintborder}%
  }
  \def\fcGraphRem#1#2{
    \tikz[remember picture,baseline=(#1.base),x=0.6cm]{#2 \faintborder}%
  }

\usepackage{mauflpic} 

\renewcommand{\phi}{\varphi}
\renewcommand{\emptyset}{\varnothing}

\newcommand{\short}[1]{}
\newcommand{\full}[1]{#1}

\DeclareMathOperator{\RO}{RO}

\newcommand{\id}{\mkern1mu\mathrm{id}}
\newcommand{\trueK}{\mathrm{true}}
\newcommand{\condtrue}{\mkern1mu\trueK}
\newcommand{\falseK}{\mathrm{false}}
\newcommand{\condfalse}{\mkern1mu\falseK}
\newcommand{\unknownK}{\mathrm{unknown}}

\newcommand{\graphf}{\textbf{Graph}_{\textbf{fin}}}
\newcommand{\graphfinj}{\textbf{Graph}_{\textbf{fin}}^{\textbf{inj}}}
\def\ILC{\mathbf{ILC}}
\def\Cospan{\mathbf{Cospan}}

\def\defeq{:=}

\newcommand{\notmodels}{\mathrel{{\mkern3mu\not\mkern-3mu\models}}}
\newcommand{\sledom}{\Relbar\joinrel\mathrel{|}}
\newcommand{\hatmodels}{\mathrel{\widehat{\rule{0ex}{1.3ex}\smash{\models}}}}

\newcommand{\catC}{\ensuremath{\mathbf{C}}}

\newcommand{\dom}{\ensuremath{\mathsf{dom}}}
\newcommand{\cod}{\ensuremath{\mathsf{cod}}}
\newcommand{\Cond}{\ensuremath{\mathbf{Cond}}}

\newcommand{\Bad}{\ensuremath{\mathsf{Bad}}\xspace}
\newcommand{\Init}{\ensuremath{\mathsf{Init}}\xspace}
\newcommand{\append}{\ensuremath{\mathsf{append}}\xspace}

\newcommand{\Arr}{\mathbf{Arr}}

\newcommand{\sem}[1]{\llbracket{#1}\rrbracket}
\newcommand{\tup}[1]{\langle{#1}\rangle}

\def\biglor{\bigvee}
\def\bigland{\bigwedge}

\newcommand{\tdots}{. \mkern1mu . \mkern1mu . \mkern1mu}

\newcommand{\HOLE}{\_}
\DeclareMathOperator{\SP}{sp}
\DeclareMathOperator{\WP}{wp}
\DeclareMathOperator{\RULE}{rule}

\newcommand{\PArr}[1]{\mathcal{P}(\Arr_{#1})}

\newcommand{\Abs}{\textrm{Abs}}
\def\obnull{0}

\newcommand{\abstractTransition}[1]{\csname ext@arrow\endcsname 0055{\csname Rightarrowfill@\endcsname}{}{#1\ }}

\makeatletter
\newcommand{\biggg}{\bBigg@{3}}
\newcommand{\Biggg}{\bBigg@{4}}
\newcommand{\bigggg}{\bBigg@{5}}
\newcommand{\Bigggg}{\bBigg@{6}}
\makeatother

\usepackage{xspace}

\makeatletter
\def\orcidID#1{\smash{\href{http://orcid.org/#1}{\protect\raisebox{-1.25pt}{\protect\includegraphics{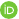}}}}}
\makeatother

\begin{document}

\title{Counterexample-Guided Abstraction Refinement for Generalized
  Graph Transformation Systems\full{ (Full Version)}}
\titlerunning{CEGAR for Generalized Graph Transformation Systems}

\author{%
  Barbara K\"onig\inst{1}\orcidID{0000-0002-4193-2889} \and
  Arend Rensink\inst{2}\orcidID{0000-0002-1714-6319} \and
  Lara Stoltenow\inst{1}\orcidID{0009-0009-1667-8573} \and
  Fabian Urrigshardt\inst{1}}
\authorrunning{B. König, A. Rensink, L. Stoltenow, F. Urrigshardt}

\institute{%
  University of Duisburg-Essen, Germany \and
  University of Twente, Netherlands \\
  \email{\{barbara\_koenig,lara.stoltenow\}@uni-due.de, arend.rensink@utwente.nl}}

\maketitle

\begin{abstract}
  This paper addresses the following verification task: Given a graph
  transformation system and a class of initial graphs, can we
  guarantee (non-)reachability of a given other class of graphs that
  characterizes bad or erroneous states?  Both initial and bad states
  are characterized by nested conditions (having first-order
  expressive power). Such systems typically have an infinite state
  space, causing the problem to be undecidable. We use abstract
  interpretation to obtain a finite approximation of that state space,
  and employ counter-example guided abstraction refinement to
  iteratively obtain suitable predicates for automated
  verification. Although our primary application is the analysis of
  graph transformation systems, we state our result in the general
  setting of reactive systems.
\end{abstract}


\section{Introduction}

One of the successful techniques to analyze systems with very large or
infinite state spaces is \emph{abstract interpretation}.  This either
under- or over-approximates the possible behaviours, reducing the
state space to a manageable size (and also covering a set of potential
initial states rather than just one) at the cost of precision:
essentially, different states are regarded as the same not if they are
for all intents and purposes equivalent, but if they are, in some
precise sense, alike. If we want to check certain properties about the
behaviour of the original system, such as the absence of errors in
reachable states or the satisfaction of more refined temporal logic
properties, this can be done on the abstracted state space instead;
however, due to the imprecision, the answers obtained in this way may
not be correct for the full state space.  In particular, if the
abstraction was an under-approximation, analysing whether an error
state is reachable may yield \emph{false negatives} (the answer
\emph{no} may be incorrect for the full, concrete state space); if it
was an over-approximation, the analysis may yield \emph{false
  positives} (the answer \emph{yes} may be incorrect).

Of these two, false positives can be detected more easily, since the answer \emph{yes} to a question of the type ``can an error state be reached'' or ``is there a trace with a certain temporal property'' comes with a \emph{witness}, being an actual trace in the abstract state space that reaches that error state or has that property. For historical reasons, such a witness is usually called a \emph{counterexample}, even though in our narrative it is rather an \emph{example}. We can then typically check whether such a counterexample actually exists in the full state space --- a task that is feasible even if the full state space is infinite. If the counterexample does exist, the original \emph{yes} answer was correct after all; if it does \emph{not} exist (in which case it is called \emph{spurious}), we know that the original \emph{yes} answer cannot be trusted; it may still be correct, but we do not know.

The method of \emph{counter-example guided
  abstraction refinement} (CEGAR)
\cite{cgjlv:abstraction-refinement-journal,hjmm:abstractions-proofs}
builds upon this principle by relying on a notion of abstraction that
can be \emph{tuned}. In particular, two states are considered to be
``alike'' if they satisfy the same predicates. A spurious
counterexample contains concrete evidence of where the original
abstraction gives rise to ``harmful'' over-approximation: namely a
trace in the abstract transition system leading to a state that does
not rule out that a certain property $\Bad$ (encoding some unwanted feature) holds,
but which does not exist in the concrete state
space. This check provides additional predicates, leading to a
stronger notion of ``alike'', which results in a new, \emph{refined} abstraction
in which at least this particular (spurious) counterexample no longer
exists. We can then start over again using the new, refined, abstract
state space, until we get either a \emph{no} answer or a \emph{yes} answer
for which the counterexample is real, or we run out of time. This
leads to the so-called CEGAR loop (see \Cref{fig:cegar-method}).
In the general case there is no guarantee
that the loop will ever terminate (e.g. due to undecidability of
the verification problem), but good results have still been
achieved in practice with the development of verification tools based
on CEGAR \cite{ccgjv:magic,bhjm:blast,blr:decade-software-mc-slam}.

A related technique for proving correctness relies on finding an \emph{invariant} $I$, which is a property that holds in the initial state and is preserved by all transitions, such that $I$ entails $\neg\Bad$. In fact, the abstraction refinement procedure outlined above may be seen as a recipe for generating such an invariant.


\begin{figure}[b]
  \centering
  \scalebox{0.54}{\input{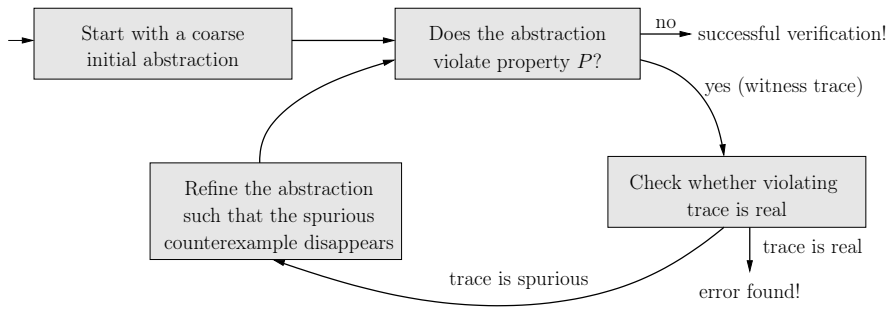}}
  \caption{A schematic depiction of the CEGAR loop}
  \label{fig:cegar-method}
\end{figure}

The contribution of this paper is to propose a
variant of CEGAR in the spirit of \emph{reactive systems}. This is a
framework developed by Milner et al.\ \cite{lm:derive-bisimulation}
that subsumes graph transformation via double-pushout. We assume that
the states are arrows in a suitable category, and transitions are
generated by conditional rules that modify those objects using
categorical operations. Our abstraction is based on so-called
\emph{nested conditions} (related to first-order logic) \cite{r:representing-fol,hp:correctness-nested-conditions}, again defined
relative to the category of choice, with a notion of satisfaction over
the states. The predicates mentioned above take the shape of such nested conditions; the
abstraction is driven by a finite set $P$ of them. The abstraction of a state, or a set of states, is defined as the subset of elements of $P$ that provably hold in that (set of) state(s). In
previous work \cite{bchk:conditional-reactive-systems}, we have shown how to compute strongest postconditions and weakest preconditions
for rules with application conditions, and this allows us to compute
the abstract (over-approximating) transitions, as well as to check whether a counterexample found on the abstract level fails to exist on the concrete level, i.e., is spurious. Given a spurious
counterexample, we can then refine the abstraction by augmenting $P$
with the characteristic properties of all the (sets of) concrete
states traversed by the counterexample. This refined abstraction is
certain not to include that counterexample any more.

A complication lies in the fact that the method outlined above at
several points requires the computation of entailment among nested
conditions. Depending on the underlying category, this is equivalent to entailment of First-Order Logic and hence undecidable; in practice, we are forced to rely on \emph{provable} entailment (i.e., using available tooling), which is necessarily a weaker relation --- or in the terminology above, may yield false negatives. Fortunately, though this introduces further imprecision in our abstraction, we show that this does not invalidate the method.

A further contribution of this paper consists in a prototype implementation of this CEGAR method for a concrete base category, viz. that of \emph{graphs}. Rules in this category are (essentially) double-pushout transformation rules, and the nested conditions correspond to ones studied before by \cite{hp:correctness-nested-conditions,r:representing-fol}. We report some experiments. Unfortunately, the performance is such that only very small examples can be analysed successfully. The biggest obstacle is the procedure for \mbox{(semi-)deciding} entailment. We believe there is a lot of room for improvement, but this is outside the scope of the current paper.

\short{A full version of the paper, including all the proofs and further
  material, is available \cite{krsu:cegar-gts-arxiv}.}

\section{Motivation}
\label{sec:motivation}

This section presents a running example that shows the principal steps
of the CEGAR method and gives an idea of its potential benefits. It is
based on a particular instantiation of our general framework
(presented formally in the next sections): the concrete states are
essentially unlabelled \emph{graphs with interfaces}
\cite{bgksz:confluence-gr-interfaces,k:hypergraph-construction-journal},
the rules are essentially \emph{double-pushout graph transformation
  rules} \cite{eps:gragra-algebraic,eept:fundamentals-agt} and our
nested conditions essentially correspond to those introduced by Habel
and Pennemann \cite{hp:correctness-nested-conditions} and Rensink
\cite{r:representing-fol}.

The example uses a graph-based representation of a set of lists, in
which the tail nodes (and no other) are marked by a self-loop. All
operations on the list are meant to guarantee that the only edges in
the graphs are those between successive list nodes and the self-loops
that mark the tail nodes; in particular, it is an error for any node
to have both an outgoing edge to a successor node and a
self-loop. This erroneous structure is captured by a condition $\Bad$;
in other words, $\Bad$ is meant to be unreachable.

We consider a scenario consisting of an initial condition $\Init_1$ that
is meant to capture precisely those graphs in which there are only
empty lists (i.e., no edges between distinct nodes), and a rule $\append$ that specifies the
extension of a list (at its tail) with a single element. Using the notation introduced in the next section, this scenario
is captured by the following conditions and rules that are based on
cospans of graphs:
\begin{align*}
\Bad & = \exists\ \emptyset \to
   \fcGraph{n1}{
      \node[gn] (n1) at (0,0) {$1$};
      \node[gn] (n2) at (1,0) {$2$};
      \draw[gedge] (n1) edge[exaloop above] (n1)
                   (n1) to (n2);
  }
  \leftarrow \fcGraph{n1}{
      \node[gn] (n1) at (0,0) {$1$};
      \node[gn] (n2) at (0.8,0) {$2$};
  } \dotTrue \\
\Init_1 & = \forall\ \emptyset \to
   \fcGraph{n1}{
      \node[gn] (n1) at (0,0) {$1$}; 
      \node[gn] (n2) at (1,0) {$2$}; 
      \draw[gedge] (n1) to (n2);
  }
  \leftarrow \fcGraph{n1}{
      \node[gn] (n1) at (0,0) {$1$};
      \node[gn] (n2) at (0.8,0) {$2$};
  }
  \dotFalse
  \\
\append & = \Big(
  \tikz[remember picture,baseline=(a.base)]{
    \node (a) at (0,0) {$\emptyset$};
    \saveBboxSouthAs{reasys motiv example L0} }
  \rightarrow
   \fcGraphRem{n1}{
      \node[gn] (n1) at (0,0) {$1$};
      \draw[gedge] (n1) edge[exaloop right] (n1);
      \saveBboxSouthAs{reasys motiv example L} }
   \leftarrow
   \fcGraphRem{n1}{
      \node[gn] (n1) at (0,0) {$1$};
      \saveBboxSouthAs{reasys motiv example I} }
   \rightarrow
   \fcGraphRem{n1}{
      \node[gn] (n1) at (0,0) {$1$};
      \node[gn] (n2) at (1,0) {$2$};
      \draw[gedge] (n1) to (n2)
                   (n2) edge[exaloop right] (n2);
      \saveBboxSouthAs{reasys motiv example R} }
   \leftarrow
      \tikz[remember picture,baseline=(a.base)]{
        \node (a) at (0,0) {$\emptyset$};
        \saveBboxSouthAs{reasys motiv example R0} }
  ,\ \condtrue
  \Big)
  \begin{tikzpicture}[remember picture, overlay]
    \draw[cospanarr] (reasys motiv example L0) to[bend right=29] node[below,pos=0.65]{$\ell$} ($(reasys motiv example I)+(-4pt,-7pt)$);
    \draw[cospanarr] (reasys motiv example R0) to[bend left=29] node[below,pos=0.65]{$r$} ($(reasys motiv example I)+(4pt,-7pt)$);
    \node[anchor=north] at (reasys motiv example L) {$L$};
    \node[anchor=north] at (reasys motiv example I) {$I$};
    \node[anchor=north] at (reasys motiv example R) {$R$};
  \end{tikzpicture}
  \\~
\end{align*}
The three graphs $L,I,R$ in rule $\append$ correspond to left-hand
side, interface, and right-hand side in double-pushout rewriting. It
should be noted that $\Init_1\models\neg\Bad$.

We want to show that $\Bad$ does not hold in any graph reachable from
any initial graph that satisfies $\Init_1$. 

The method works by initially assuming that the only knowledge we have
about a state is whether $\Init_1$, $\Bad$ or the negation of either
holds there, and checking if that knowledge is sufficient to show that
$\neg\Bad$ holds everywhere. (In fact, CEGAR can be seen as a way to
automatically generate invariants or -- more generally -- conditions
that are guaranteed to hold at certain execution points.)

Let us denote $P=\{\Init_1, \Bad\}$; then abstract states are elements
of $\{\trueK,\allowbreak\falseK,\allowbreak\unknownK\}^{P}$ --- or, equivalently, subsets of
$P\cup \{\neg \mathcal{P} \mid \mathcal{P}\in P\}$ (or conjunctions of
such predicates) that do not include both $\mathcal{P}$ and
$\neg \mathcal{P}$ for any $\mathcal{P}\in P$. The method involves the
following steps.
\begin{enumerate}
\item Compute, for every rule, the strongest postcondition ($\SP$) for the
  next unexplored abstract state. In our example, the initial state
  consists of\footnote{Note that $\Init_1\models\lnot\Bad$ implies
    $\Init_1\equiv\Init_1\land\lnot \Bad$, hence we could start with
    either formula.} $\{\Init_1,\lnot\Bad\}$ and hence we start by
  computing $\SP(\Init_1\wedge\neg\Bad,\append)$, which yields
  ${\cal A}= \exists\ \emptyset \to \fcGraph{n1}{ \node[gn] (n1) at
    (0,0) {$1$}; \node[gn] (n2) at (1,0)
    {$2$}; \draw[gedge] (n1) to (n2) (n2) edge[exaloop right] (n2); }
  \leftarrow \fcGraph{n1}{ \node[gn] (n1) at (0,0)
    {$1$}; } \ . \mathcal A'$, where the subcondition
  $\mathcal A'$ guarantees that $\Init_1 \land \neg\Bad$ holds for
  nodes $1$ and $2$, as well any additional nodes that already existed
  in the graph.

\item Infer which of the elements of $P$ or their negations are
  entailed by $\cal A$. In our example, we only have
  ${\cal A}\models \neg\Init_1$: we cannot infer either
  ${\cal A}\models \Bad$ or ${\cal A}\models \neg\Bad$. It follows
  that the successor state is $\{\neg \Init_1\}$. This gives rise to the following transition system:
\[ \{\Init_1,\neg\Bad\} \abstractTransition{\append} \{\neg\Init_1\} \]

\item Repeat the previous two steps for every new state, until there
  are no more states to be found (note that the number of reachable
  states is bounded by $3^{|P|}$) or we find a state $s$ for which
  $\neg\Bad\not\in s$. In the former case we are done: the system is
  guaranteed to be error-free. In the latter case, the system has a
  potentially faulty behaviour: namely, the trace from the initial
  state to $s$. Such a trace is called a
  \emph{counterexample}. In our running example, we already have such a
  counterexample, viz.\ the trace consisting of a single application of
  $\append$.

\item Check whether this trace really represents faulty behaviour, by
  computing the weakest precondition $\mathcal W_1$ for $\neg \Bad$.  If
  $\Init_1\models \mathcal W_1$, the counterexample is \emph{spurious}; if
  not, then any graph satisfying $\Init_1 \land \neg\mathcal W_1$, when subjected to the successive
  rules of the counterexample, will indeed give rise to a concrete
  state satisfying $\Bad$ --- hence we have a real error. In our
  running example, the weakest precondition $\WP(\append,\neg\Bad)$ is
  (simplified and in abbreviated notation):
  \[
  \mathcal W_1 \equiv
    \forall\ \fcGraph{n1}{
      \node[gn] (n1) at (0,0) {$1$};
      \draw[gedge] (n1) edge[exaloop right] (n1)
                   (n1) edge[exaloop left] (n1);
    } \dotFalse \land
    \forall\ \fcGraph{n1}{
      \node[gn] (n1) at (0,0) {$1$};
      \node[gn] (n2) at (1,0) {$2$};
      \node[gn] (n3) at (2,0) {$3$};
      \draw[gedge] (n1) edge[exaloop above] (n1)
                   (n2) edge[exaloop above] (n2)
                   (n2) to (n3);
    } \dotFalse \land
    \forall\ \fcGraph{n1}{
      \node[gn] (n0) at (0,0) {$0$};
      \node[gn] (n1) at (1,0) {$1$};
      \draw[gedge] (n1) edge[exaloop above] (n1)
                   (n0) edge[exaloop above] (n0)
                   (n0) to (n1);
    } \dotFalse
  \enspace,
  \]
  which is not entailed by $\Init_1$: for instance, 
  $\fcGraph{n1}{ \node[gn] (n1) at (0,0) {$\phantom1$}; \draw[gedge]
    (n1) edge[exaloop left] (n1) (n1) edge[exaloop right] (n1); }$ satisfies
  $\Init_1 \land \neg\mathcal W_1$. From it we can reach a bad state in
  one \append-step, hence the system is erroneous.
\end{enumerate} 
We see that the method has uncovered the fact that our initial condition $\Init_1$ is not strong enough to guarantee $\neg\Bad$: we had not considered that nodes may have more than one self-loop. We can repair this by strengthening $\Init_1$ so as to rule this out. Let us redo the analysis on the basis of
\[
\Init_2 = \Init_1
  \land \forall\ \emptyset \to \fcGraph{n1}{
      \node[gn] (n1) at (0,0) {$1$};
      \draw[gedge] (n1) edge[exaloop right] (n1);
      \draw[gedge] (n1) edge[exaloop left] (n1);
  }
  \leftarrow \fcGraph{n1}{
      \node[gn] (n1) at (0,0) {$1$};
  }
  \dotFalse
\]
Note that, on its own, $\neg\Bad$ is not an invariant: it is
\emph{not} the case that the result of applying \append to a graph
satisfying $\neg\Bad$ will certainly also satisfy $\neg\Bad$ ---
otherwise the problem would be easier. Similarly, $\Init_2$ is not an
invariant.

After repeating the first two steps above, we once more find that there is a reachable state not containing $\neg\Bad$, identifying a counterexample, this time consisting of two rule applications:
\[ \{\Init_2,\neg\Bad\} \abstractTransition{\append} \{\neg\Init_2,\neg\Bad\} \abstractTransition{\append} \{\neg\Init_2\} \]
The weakest precondition computation gives us two conditions
${\cal W}_1=\WP(\append,\allowbreak\neg \Bad)$ (identical to the one
above) and ${\cal W}_2=\WP(\append,{\cal W}_1)$ (in this particular
case $\mathcal W_2 \equiv \mathcal W_1$). Now
$\Init_2 \models {\cal W}_2$, and hence the counterexample is
spurious. When that happens, the method continues as follows:
\begin{enumerate}[resume]
\item Refine $P$ by adding predicates that ensure the counterexample
  no longer occurs on the abstract states. In particular, augmenting
  $P$ with the weakest preconditions that we just computed in order to
  check for spuriousness --- in our case ${\cal W}_1$ --- will do the
  trick. (We could also choose to add individual conjuncts
  ${\cal W}_{ij}$ if
  ${\cal W}_i={\cal W}_{i1}\wedge \cdots \wedge {\cal W}_{in}$, to possibly obtain a better abstraction. In our
  running example, the three subconditions of $\mathcal W_1$ could be turned into three new predicates.)
  We have the guarantee that this will indeed eliminate
  the spurious counterexample.

\item Restart the analysis at Step~1 on the basis of the refined $P$.
  Repeat the process until either a real counterexample is found
  (as happened in our first iteration), all reachable (abstract)
  states have been computed and they all entail $\lnot\Bad$, or time
  is up. For our running example, the second case occurs:
\[ \{\Init_2,\neg\Bad,\mathcal W_1\} \abstractTransition{\append} \{\neg\Init_2,\neg\Bad,\mathcal W_1\} 
\!\!\tikz[baseline=(n.south)]{\node (n) at (0,0) {}; \draw[double equal sign distance,-{Implies}] (n) to[out=20,in=-20,looseness=12] node[right] {\small\append} (n);}
\]
\end{enumerate}
The running example shows the power of the analysis method: it allows
us to prove that, for all graphs satisfying a given initial condition
(of which there are infinitely many), for all possible sequences of
rule applications (of arbitrary length, and generating graphs of
unbounded size), the outcome satisfies the well-formedness condition
embodied in $\neg\Bad$; or, if that is not the case, to find a
counterexample. Of course, there are practical limitations, which we
will discuss in \Cref{sec:undecidability-issues}.

\section{Preliminaries}
\label{sec:preliminaries}

\subsection{Abstract interpretation}
\label{sec:abstract-interpretation}

We rely on the principles of the theory of abstract interpretation
\cite{c:abstract-interpretation,cc:ai-unified-lattice-model}, based
on lattices and Galois connections. We first recall the definitions.

A \emph{complete lattice} $(\mathbb{C},\sqsubseteq)$ consists of a set
$\mathbb{C}$ with a partial order $\sqsubseteq$ such that each
$Y\subseteq \mathbb{C}$ has a least upper bound $\bigsqcup Y$ (also
called supremum, join) and a greatest lower bound $\bigsqcap Y$ (also
called infimum, meet).

Let $\mathbb{C}$, $\mathbb{A}$ be two lattices.  A \emph{Galois
  connection} from $\mathbb{C}$ to $\mathbb{A}$ is a pair
$\alpha\colon \mathbb{C}\to\mathbb{A}$,
$\gamma\colon \mathbb{A}\to\mathbb{C}$ of monotone functions, such that for all $\ell\in{\mathbb C}$:
$\ell\sqsubseteq \gamma(\alpha(\ell))$ and for all $m\in \mathbb A$:
$\alpha(\gamma(m)) \sqsubseteq m$.
%
Intuitively $\mathbb{C}$ represents (more) concrete values and
$\mathbb{A}$ (more) abstract values, which are connected by the
abstraction $\alpha$ and concretization $\gamma$. The order
indicates whether the values are more or less precise: i.e., whenever
$a\sqsubseteq b$, then $a$ is supposed to provide higher precision
than $b$. The function $\alpha$
(resp.\ $\gamma$) is also called the \emph{left}
(resp.\ \emph{right}) \emph{adjoint}.

Given a function $f\colon \mathbb{C}\to\mathbb{C}$ on the concrete
values, we say that $f^\#\colon \mathbb{A}\to \mathbb{A}$ is an
over-approximation of $f$ whenever $\alpha\circ f\circ \gamma
\sqsubseteq f^\#$ (pointwise). Whenever equality holds $f^\#$ is the induced
over-approximation of $f$.

\subsection{Categories}

We will use standard concepts from category theory. Given an arrow
$f\colon A\to B$, we write $\dom(f)=A$, $\cod(f)=B$. For two arrows
$f \colon A \to B$, $g \colon B \to C$ we denote their composition by
${f;g} \colon A \to C$.
%

We will state our results in a general framework, allowing for easy generalization of our results to other applications.
An important type of category that we will focus on are cospan categories, which are particularly useful for reactive systems (to be defined later).
Given a base category $\mathbf D$ with
pushouts, the category $\Cospan(\mathbf D)$ has as objects the objects
of $\mathbf D$ and as arrows \emph{cospans}, which are equivalence
classes of pairs of arrows of the form
$\smash{A \xrightarrow{f_L} X \xleftarrow{f_R} B}$, where the middle
object is considered up to isomorphism. Cospan composition is
performed via pushouts\full{ (for details see Appendix~\ref{app:cospans})}.

A~cospan is \emph{left-linear} if its left leg $f_L$ is a
mono.  For adhesive categories~\cite{ls:adhesive-journal}, the
composition of left-linear cospans again yields a left-linear
cospan. $\ILC(\mathbf D)$ will denote the
subcategory of $\Cospan(\mathbf D)$ where the arrows are restricted to
left-linear cospans (historically called \emph{input}-linear; hence $\ILC$.)

Our running example is based on the category $\mathbf{D}=\graphf$, which has finite graphs as objects and graph morphisms as arrows.

\nosectionappendix
\begin{toappendix}
  \section{Additional Material for \S\ref{sec:preliminaries}\ (Preliminaries)}
  \paragraph*{Graphs and graph morphisms}

  We will define in more detail which graphs and graph morphisms we
  are using: in particular, a graph is a tuple $G = (V,E,s,t,\ell)$,
  where $V,E$ are sets of nodes respectively edges, $s,t\colon E\to V$
  are the source and target functions and $\ell\colon V\to \Lambda$
  (where $\Lambda$ is a set of labels) is the node labelling
  function. In the examples we will always omit node labels by
  assuming that there is only a single label.

  A graph $G$ is finite if both $V$ and $E$ are finite.

  Furthermore, given two graphs $G_i = (V_i,E_i,s_i,t_i,\ell_i)$,
  $i\in\{1,2\}$, a graph morphism $\phi\colon G_1\to G_2$ consists of
  two maps $\phi_V\colon V_1\to V_2$, $\phi_E\colon E_1\to E_2$ such
  that $\phi_V\circ s_1 = s_2\circ \phi_E$, $\phi_V\circ t_1 =
  t_2\circ \phi_E$ and $\ell_1 = \ell_2\circ\phi_V$.

  In the examples, the mapping of a morphism is given implicitly by
  the node identifiers: for instance, $
    \fcGraph{n1}{\node[gn] (n1) at (0,0) {$1$}; \node[gn] (n2) at (1,0) {$2$};}
    \rightarrow \fcGraph{n1}{
      \node[gn] (n1) at (0,0) {$1$};
      \node[gn] (n2) at (2,0) {$2$};
      \node[gn] (n3) at (1,0) {$3$};
      \draw[gedge] (n1) to (n3);
      \draw[gedge] (n2) to (n3);
    }
  $ adds the node identified by $3$ and adds two edges from the existing
  nodes identified by $1$ and $2$.

  \paragraph*{Cospans and cospan composition}
  \label{app:cospans}
  We compose two cospans \mbox{$%
    f \colon A \xrightarrow{f_L} X \xleftarrow{f_R} B$,} \mbox{$%
    g \colon B \xrightarrow{g_L} Y \xleftarrow{g_R} C$}
  by taking the pushout $(p_L, p_R)$ of $(f_R, g_L)$ as shown in
  \Cref{fig-cospan-compo}.  The result is the cospan
  \mbox{${f;g} \colon A \xrightarrow{f_L;p_L} Z \xleftarrow{g_R;p_R}
    C$}, where $Z$ is the pushout object of $f_R,\; g_L$.  We see an
  arrow $f \colon A \to C$ of $\Cospan(\mathbf D)$ as an object~$B$
  of~$\mathbf D$ equipped with two interfaces $A,C$ and corresponding
  arrows $f_L,f_R$ to relate the interfaces to $B$, and composition
  glues the inner objects of two cospans via their common
  interface.

  \begin{figure}[h]
    \centering\begin{tikzpicture}[x=1.10cm,y=1.10cm]

      \def\sqwo{1.2} \def\sqoo{1.4} \def\sqho{0.95}

      \node (a) at (-\sqwo-\sqoo,0) {$A$}; \node (x) at (-\sqwo,0)
      {$X$}; \node (b) at (0,\sqho) {$B$}; \node (y) at (\sqwo,0)
      {$Y$}; \node (c) at (\sqwo+\sqoo,0) {$C$}; \node (z) at
      (0,-\sqho) {$Z$};

      \def\smf{\footnotesize} \draw[cospanint] (a) --
      node[above]{\smf$f_L$} (x); \draw[cospanint] (b) --
      node[above,pos=0.7]{\smf$f_R\ $} (x); \draw[cospanint] (b) --
      node[above,pos=0.7]{\smf\kern5pt$g_L$} (y); \draw[cospanint] (c)
      -- node[above]{\smf$g_R$} (y);

      \draw[cospanint] (x) -- node[below,pos=0.4]{\smf$p_L$\kern8pt}
      (z); \draw[cospanint] (y) --
      node[below,pos=0.4]{\smf\kern12pt$p_R$} (z);

      \draw[cospanarr] (a) to[bend left=30] node[above,pos=0.2]{$f$}
      (b); \draw[cospanarr] (b) to[bend left=30]
      node[above,pos=0.8]{$\ g$} (c); \draw[cospanarr] (a) to[bend
      right=50] node[below]{$f;g$} (c);

      \node at (0,0) {\textup{(PO)}};

    \end{tikzpicture}%
    \caption{Composition of cospans $f$ and~$g$ is done via pushouts}%
    \label{fig-cospan-compo}%
  \end{figure}
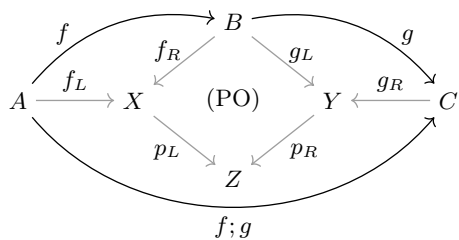
  In order to make sure that arrow composition in $\Cospan(\mathbf D)$
  is associative on the nose, we quotient cospans up to
  isomorphism. In more detail: two cospans
  $f \colon A \xrightarrow{f_L} X \xleftarrow{f_R} B$,
  $g \colon A \xrightarrow{g_L} Y \xleftarrow{g_R} B$ are equivalent
  whenever there exists an iso $\iota\colon X\to Y$ such that
  $f_L;\iota = g_L$, $f_R;\iota = g_R$. Then, arrows are equivalence
  classes of cospans.
\end{toappendix}



\subsection{Generalized Conditions}
\label{sec:reactive-systems-conditions}

As in previous work~\cite{sksclo:coinductive-satisfiability-nested-conditions}
we consider nested conditions --- from here on just called \emph{conditions} --- over an arbitrary category $\catC$ in
the spirit of reactive systems
\cite{lm:derive-bisimulation,l:congruences-reactive}.  Following
\cite{r:representing-fol,hp:correctness-nested-conditions}, we define
conditions as finite tree-like structures, where nodes are annotated
with quantifiers and objects, and edges are annotated with arrows.

\begin{definition}[Condition]
  \label{def:condition-new}
  Let $\catC$ be a category. A condition $\mathcal{A}$ over an object
  $A$ in $\catC$ is defined inductively as follows: it is either
  \begin{itemize}
    \item
      a finite conjunction of universals
      $\bigland\nolimits_{i \in \{1,\tdots,n\}} \forall f_i.\mathcal A_i
      = \forall f_1.\mathcal A_1 \land \tdots \land \forall f_n.\mathcal A_n$,~or
    \item
      a finite disjunction of existentials
      $\biglor\nolimits_{i \in \{1,\tdots,n\}} \exists f_i.\mathcal A_i
      = \exists f_1.\mathcal A_1 \lor \tdots \lor \exists f_n.\mathcal A_n$
  \end{itemize}
  where $f_i\colon A\to A_i$ are arrows in $\catC$ and
  $\mathcal A_i$ are conditions over $A_i$. We call
  $A = \RO(\mathcal A)$ the \emph{root object} of the condition
  $\mathcal A$.  Each subcondition $\mathcal Q f_i.\mathcal A_i$
  ($\mathcal Q \in \{\forall,\exists\}$) is called a \emph{child
    of~$\mathcal A$}.  The constants $\condtrue_A$ (empty conjunction) and
  $\condfalse_A$ (empty disjunction) serve as the base cases.
  We will omit subscripts in $\condtrue_A$ and $\condfalse_A$ when clear
  from the context. 
  The set
  of all conditions over $A$ is denoted by $\Cond_A$, and
  $\Arr_A$
  refers to the $A$-sourced arrows (i.e., potential models) of $\catC$.

\end{definition}
Instantiated with $\catC=\graphf$, conditions are equivalent to graph
conditions as defined in \cite{hp:correctness-nested-conditions}, and
equivalence to first-order logic has been shown in
\cite{r:representing-fol}.  Cospan conditions (with
$\catC=\ILC(\graphf)$) have previously been used
\cite{lmcs:8951,sksclo:coinductive-satisfiability-nested-conditions}.
Standard graph conditions can trivially be encoded into cospan
conditions, and cospan conditions can be translated to equivalent
graph conditions.

Intuitively, 
conditions check for the occurrence of certain subgraphs or patterns for which the context satisfies a child condition.
For instance, the cospan condition
\mbox{%
$\forall\ \emptyset \to
\fcGraph{n1}{\node[gn] (n1) at (0,0) {$1$}; \node[gn] (n2) at (1,0) {$2$}; \draw[gedge] (n1) to (n2); }
  \leftarrow \fcGraph{n1}{\node[gn] (n1) at (0,0) {$1$}; \node[gn] (n2) at (0.7,0) {$2$}; }
\dotEx
\fcGraph{n1}{\node[gn] (n1) at (0,0) {$1$}; \node[gn] (n2) at (0.7,0) {$2$}; }
\to
\fcGraph{n1}{\node[gn] (n1) at (0,0) {$1$}; \node[gn] (n2) at (1,0) {$2$}; \draw[gedge] (n2) to[bend left=10] (n1); }
\leftarrow \fcGraph{n1}{\node[gn] (n1) at (0,0) {$1$}; \node[gn] (n2) at (0.7,0) {$2$}; }
\dotTrue$}
requires that for every edge, a second edge in the reverse direction also exists.
For additional examples and discussion we refer to \cite{lmcs:8951}.

To be consistent with \cite{sksclo:coinductive-satisfiability-nested-conditions},
\Cref{def:condition-new} restricts conjunction to
universal and disjunction to existential subformulas; e.g.,
$\exists f.\mathcal A \land \exists g.\mathcal B$ is excluded.
However, conditions that violate this syntactic restriction can easily be rewritten ---
e.g., to $\forall \id.\exists f.\mathcal A \land \forall \id.\exists g.\mathcal B$ for the above example. Hence, in examples we sometimes write
$\mathcal{A}\land \mathcal{B}$ or $\mathcal{A}\lor \mathcal{B}$ for
arbitrary conditions.

\begin{definition}[Satisfaction]
  \label{def:satisfaction}
  Let $\mathcal A \in \Cond_A$ and  let $a \colon A \to B$ be an arrow.
  We define the satisfaction relation
  $a \models \mathcal A$ as follows:
  \begin{itemize}
  \item $a \models \bigland_{i\in I} \forall f_i.\mathcal A_i$
    iff for every  $i\in I$ and every arrow
    $g \colon \RO(\mathcal A_i) \to B$
    we have: if~$a = f_i;g$, then $g \models \mathcal A_i$.
  \item $a \models \biglor_i \exists f_i.\mathcal A_i$ iff there
    exists $i\in I$ and an arrow $g \colon \RO(\mathcal A_i) \to B$
    such that $a = f_i;g$ and $g \models \mathcal A_i$.
  \end{itemize}
\end{definition}
We define a concretization
  $\sem\_\colon\Cond_A\to\PArr{A}$ (for arbitrary $A$) via
  $\sem{\mathcal{A}} = \{x\in \Arr_{\RO(\mathcal A)}\mid x \models
  \mathcal{A}\}$, mapping conditions to the set of arrows that satisfy them.
From the above it follows that $\sem{\condtrue_A}=\Arr_A$ and $\sem{\condfalse_A}=\emptyset$.

We write $\mathcal A \models \mathcal B$ ($\mathcal{A}$ entails
$\mathcal{B}$) if $\RO(\mathcal A) = \RO(\mathcal B)$ and for every
arrow $a\in\Arr_{\RO(\mathcal A)}$ we
have: if $a \models \mathcal A$, then
$a \models \mathcal B$. We write $\mathcal A \equiv \mathcal B$ ($\mathcal A$ and $\mathcal B$ are equivalent) if $\mathcal A \models \mathcal B$
and $\mathcal B \models \mathcal A$.

Since conditions are equivalent to first-order logic
\cite{r:representing-fol} for $\catC=\graphf$, the
satisfiability, entailment and equivalence problems are undecidable,
but semi-decidable. In fact, in
\cite{sksclo:coinductive-satisfiability-nested-conditions} we have
provided a semi-decision procedure for satisfiability in the general
case, based on a predecessor technique for graph conditions
\cite{lo:tableau-graph-properties}.

\subsection{Conditional reactive systems}
\label{sec:reactive-systems}

We now define conditional reactive systems, which were introduced in
\cite{lm:derive-bisimulation} and extended with
application conditions in \cite{DBC-CRS}.   In our
definition, we closely follow \cite{lmcs:8951}.
We fix a distinguished object $0$ (not necessarily the initial object
in the category).

\begin{definition}[Reactive system rules]\label{def-reasys}%
  Let $\catC$ be a category with a distinguished object $\obnull$
  (not necessarily initial).  A \emph{rule} $\mathcal R=(\ell,r,\mathcal C)$ consists of
  arrows $\ell,r \colon {\obnull \to I}$ (called left-hand side and
  right-hand side) and a condition $\mathcal C$ with root object $I$.  A \emph{reactive system} is a set of rules.

  Let $\mathcal S$ be a reactive system and
  $a, a' \colon \obnull \to J$ be arrows. We say that $a$
  \emph{reduces to} $a'$ ($a \leadsto a'$) in $\mathcal{S}$ whenever there exists a
  rule $(\ell,r,\mathcal C) \in \mathcal S$ 
  and an arrow $c \colon I \to J$ (the \emph{reactive context}) such
  that $a = \ell;c,\ a' = r;c$ and $c \models \mathcal C$.
\end{definition}
\noindent
Note that $\cal C$ is not a pre- or post-application condition, but is
specified over the context in which the reaction takes place. Reactive systems instantiated with cospans
(where $0$ is the empty graphs) \cite{lmcs:8951,ss:reactive-cospans,s:deriving-congruences} yield
exactly double-pushout rewriting \cite{eps:gragra-algebraic}, hence reactive systems over $\ILC(\graphf)$ essentially describe DPO graph transformation systems with monic matching.

\begin{example}
  The reactive system over $\ILC(\graphf)$ having a single rule $\mathcal R$ (defined below) adds edges between arbitrary nodes, but only if such an edge does not already exist:
  \vskip-2pt

  \begin{align*}
    \mathcal R &= \Big(
      \tikz[remember picture,baseline=(a.base)]{
        \node (a) at (0,0) {$\emptyset$}; 
        \saveBboxNorthAs{reasys first example L0} }
        \rightarrow 
        \fcGraphRem{n1}{
          \node[gn] (n1) at (0,0) {$1$}; 
          \node[gn] (n2) at (1,0) {$2$}; 
          \saveBboxNorthAs{reasys first example L} }
        \leftarrow 
        \fcGraphRem{n1}{
          \node[gn] (n1) at (0,0) {$1$}; 
          \node[gn] (n2) at (1,0) {$2$}; 
          \saveBboxNorthAs{reasys first example I}  }
        \rightarrow 
        \fcGraphRem{n1}{
          \node[gn] (n1) at (0,0) {$1$}; 
          \node[gn] (n2) at (1,0) {$2$}; 
          \draw[gedge] (n1) to (n2); 
          \saveBboxNorthAs{reasys first example R}  }
      \leftarrow 
      \tikz[remember picture,baseline=(a.base)]{
        \node (a) at (0,0) {$\emptyset$}; 
        \saveBboxNorthAs{reasys first example R0} }
      ,\ \mathcal C
      \Big) \\
    \mathcal C &=
      \forall\ \fcGraph{n1}{\node[gn] (n1) at (0,0) {$1$}; \node[gn] (n2) at (1,0) {$2$}; }
      \rightarrow \fcGraph{n1}{\node[gn] (n1) at (0,0) {$1$}; \node[gn] (n2) at (1,0) {$2$}; \draw[gedge] (n1) to (n2); }
      \leftarrow \fcGraph{n1}{\node[gn] (n1) at (0,0) {$1$}; \node[gn] (n2) at (1,0) {$2$}; }
      \dotFalse
    \begin{tikzpicture}[remember picture, overlay]
      \draw[cospanarr] (reasys first example L0) to[bend left=27] node[above,pos=0.65]{$\ell$} ($(reasys first example I)+(-4pt,7pt)$);
      \draw[cospanarr] (reasys first example R0) to[bend right=27] node[above,pos=0.65]{$r$} ($(reasys first example I)+(4pt,7pt)$);
      \node[anchor=south] at (reasys first example L) {$L$};
      \node[anchor=south] at (reasys first example I) {$I$};
      \node[anchor=south] at (reasys first example R) {$R$};
    \end{tikzpicture}
  \end{align*}
\vskip-1.45\baselineskip
\end{example}

\subsection{Shift operation and Hoare logic}
\label{sec:shift-hoare}

Nested conditions are equipped with a shift operation. More
concretely, given $\mathcal{A}\in \Cond_A$, $c\colon A\to B$,
$\mathcal A_{\downarrow c}\in \Cond_B$ can be understood as a partial
evaluation of $\mathcal A$ under the assumption that an arrow $c$ is
already ``present''.  In particular,
it is defined as
$d \models \mathcal{A}_{\downarrow c} :\!\iff c;d \models
\mathcal{A}$. Here we do not delve into details, but just remark that
a shift can be computed via so-called representative squares\full{ (for
further information see the appendix)}.

\begin{toappendix}
We will now define the notion of representative squares, which
describe representative ways to close a span of arrows. They
generalize idem pushouts~\cite{lm:derive-bisimulation} and
borrowed context diagrams~\cite{ek:congruence-dpo-journal}.
They are needed to define the shift operation and subsequently the
construction of weakest preconditions and strongest postconditions.

\begin{definition}[Representative squares~\cite{bchk:conditional-reactive-systems}]\label{def:representative-squares}%
\floatingpicspaceright[4]{2.9cm}
\begin{floatingpicright}{2.8cm}%
  \hfill\begin{tikzpicture}[x=0.9cm,y=0.9cm,baseline=2pt]
    \node (a) at (0,0) {$A$};
    \node (b) at (1.5,0) {$B$};
    \node (c) at (0,-1.7) {$C$};
    \node (d) at (1.1,-1.3) {$D$};
    \node (ds) at (1.9,-2.0) {$D'$};
    \draw[->] (a) -- node[above,overlay]{$\alpha_1$} (b);
    \draw[->] (a) -- node[left]{$\alpha_2$} (c);

    \draw[->,fancydotted] (b) -- node[left,pos=0.3]{$\beta_1$} (d);
    \draw[->,fancydotted] (c) -- node[above,pos=0.35]{$\beta_2$} (d);

    \draw[->,fancydotted] ($(d.center)+(5pt,-5pt)$) -- ($(ds.center)+(-5pt,5pt)$);
    \node at ($(d.center)+(11pt,-4pt)$) {$\gamma$};
    \draw[->] (b) -- node[right]{$\delta_1$} (ds);
    \draw[->] (c) -- node[below,overlay]{$\delta_2$} (ds);
  \end{tikzpicture}
\end{floatingpicright}
A class $\kappa$ of commuting squares in a category $\catC$ is
\emph{represen\-tative}
if 
for every commuting square $\alpha_1 ; \delta_1 = \alpha_2 ; \delta_2$
in $\catC$ there exists a representative square
$\alpha_1 ; \beta_1 = \alpha_2 ; \beta_2$ in $\kappa$ and an arrow
$\gamma$ such that $\delta_1 = \beta_1;\gamma$ and
$\delta_2 = \beta_2;\gamma$.

  \floatingpicspaceright[1]{2.9cm} For two arrows
  $\alpha_1 \colon A \to B,\ \alpha_2 \colon A \to C$, we define
  $\kappa(\alpha_1,\alpha_2)$ as the set of pairs of arrows
  $(\beta_1,\beta_2)$ which, together with $\alpha_1,\alpha_2$, form
  representative squares in $\kappa$.
\end{definition}

Compared to weak pushouts, more than one square might be needed to represent all commuting squares that extend a given span $(\alpha_1,\alpha_2)$.
In categories with pushouts (such as $\graphf$), pushouts are the most
natural candidate for representative squares.  In $\graphfinj$
pushouts do not exist, but jointly epi squares can be used
instead. For cospan categories, one can use borrowed context
diagrams~\cite{ek:congruence-dpo-journal}\full{ (see
Appendix~\ref{app:borrowed-context} for a summary)}.

For many categories of interest -- such as $\graphf$ and
$\ILC(\graphf)$ -- we can guarantee a choice of $\kappa$ such that
each set $\kappa(\alpha_1,\alpha_2)$ is finite and computable.
In the rest of this paper, we assume that we work in such
a category, and use such a class $\kappa$. Hence the
constructions described below are effective since the finiteness of
the transformed conditions is preserved.

\paragraph*{Borrowed context diagrams}
  \label{app:borrowed-context}

For cospan categories over adhesive categories (such as $\ILC(\graphf)$), borrowed context
diagrams -- initially introduced as an extension of DPO
rewriting~\cite{ek:congruence-dpo-journal} -- can be used as representative squares. Before
we can introduce such diagrams, we first need the notion of jointly
epi.

\begin{definition}[Jointly epi]
  A pair of arrows $f \colon B \to D,\ g \colon C \to D$ is
  \emph{jointly epi} (\emph{JE}) if for each pair of arrows
  $d_1, d_2 \colon D \to E$ the following holds: if $f;d_1 = f;d_2$
  and $g;d_1 = g;d_2$, then $d_1 = d_2$.
\end{definition}

In $\graphf$ jointly epi equals jointly surjective, meaning
that each node or edge of $D$ is required to have a preimage under $f$
or $g$ or both ($D$ contains only images of $B$ or~$C$).

This criterion is similar to, but weaker than a pushout:
For jointly epi morphisms $d_1 \colon B \to D,\ d_2 \colon C \to D$,
there are no restrictions on which elements of $B,C$ can be merged in~$D$.
However, in a pushout constructed from morphisms $a_1 \colon A \to B,\ a_2 \colon A \to C$,
elements in~$D$ can (and must) only be merged if they have a common preimage in $A$.
(Hence every pushout generates a pair of jointly epi arrows, but not vice versa.)

\begin{definition}[Borrowed context diagram~\cite{DBC-CRS}]
  A commuting diagram in the category $\ILC(\catC)$, where $\catC$ is
  adhesive, is a \emph{borrowed context diagram} whenever it has the
  form of the diagram shown in \Cref{fig-bc-diag}, and the four
  squares in the base category $\catC$ are pushout (PO), pullback (PB)
  or jointly~epi (JE) as indicated. Arrows depicted as
  $\rightarrowtail$ are mono. In particular, $L\rightarrowtail G^+$,
  $G\rightarrowtail G^+$ must be jointly epi.
\end{definition}

\Cref{fig-bc-venn} shows a more concrete version of
\Cref{fig-bc-diag}, where graphs and their overlaps are depicted by
Venn diagrams (assuming that all morphisms are injective). Because of
the two pushout squares, this diagram can be interpreted as
composition of cospans $a;f = \ell;c = D \rightarrow G^+ \leftarrow K$
with extra conditions on the top left and the bottom right square.
The top left square fixes an overlap $G^+$ of $L$ and $G$, while $D$
is contained in the intersection of $L$ and $G$ (shown as a hatched
area). Being jointly epi ensures that it really is an overlap and does
not contain unrelated elements.
The top right pushout corresponds to the left
pushout of a DPO rewriting diagram. It contains a total match of $L$
in $G^+$.  Then, the bottom left pushout gives us the minimal borrowed
context $F$ such that applying the rule becomes possible.
The top left and the bottom left squares together ensure that the contexts to be considered are not larger than necessary.
The bottom right pullback ensures that the interface $K$ is as large
as possible.

For more concrete examples of borrowed context diagrams, we refer to~\cite{ek:congruence-dpo-journal,lmcs:8951}.

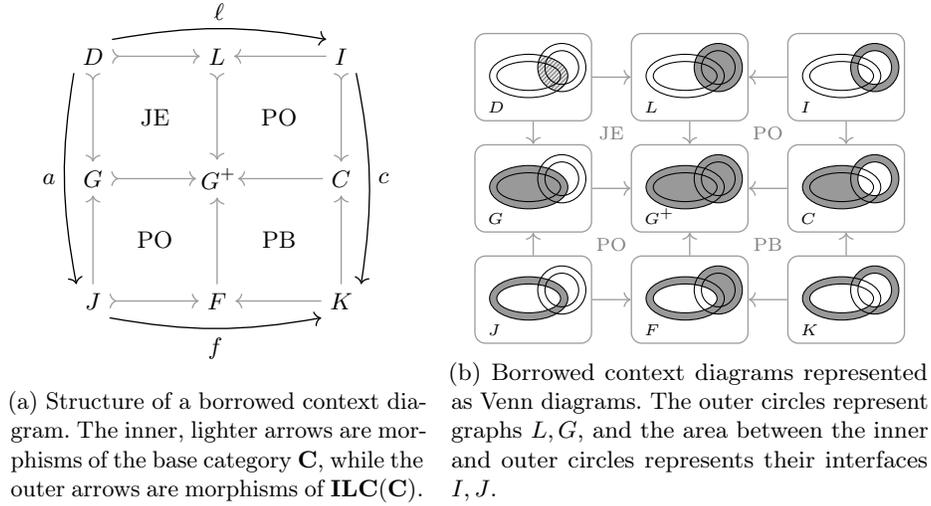
\begin{figure}[t]
  \begin{subfigure}[b]{0.45\textwidth}
    \centering\begin{tikzpicture}[x=1.10cm,y=1.10cm]
      \def\sqw{1.5}
      \def\sqh{1.48}
      \pgfmathsetmacro\sqww{2*\sqw}
      \pgfmathsetmacro\sqhh{2*\sqh}

      \path[use as bounding box] (-0.7,0.8) rectangle (2*\sqw+0.7,-2*\sqh-0.87);

      \node (d) at (0,0) {$D$};
      \node (l) at (\sqw,0) {$L$};
      \node (i) at (\sqww,0) {$I$};
      \node (g) at (0,-\sqh) {$G$};
      \node (gp) at (\sqw,-\sqh) {$\mkern-3mu G^+\mkern-5mu$};
      \node (c) at (\sqww,-\sqh) {$C$};
      \node (j) at (0,-\sqhh) {$J$};
      \node (f) at (\sqw,-\sqhh) {$F$};
      \node (k) at (\sqww,-\sqhh) {$K$};

      \def\sqlfs{\footnotesize}
      \node at (0.5*\sqw,-0.5*\sqh) {\sqlfs\textup{JE}};
      \node at (1.5*\sqw,-0.5*\sqh) {\sqlfs\textup{PO}};
      \node at (0.5*\sqw,-1.5*\sqh) {\sqlfs\textup{PO}};
      \node at (1.5*\sqw,-1.5*\sqh) {\sqlfs\textup{PB}};

      \draw[cospanintmono]
        (d) edge (l) edge (g)
        (l) edge (gp)
        (g) edge (gp)
        (i) edge (c)
        (j) edge (f);
      \draw[cospanint]
        (j) edge (g)
        (k) edge (f) edge (c)
        (f) edge (gp)
        (i) edge (l)
        (c) edge (gp);

      \draw[cospanarr] (d.north east) to[bend left =10] node[above]{$\ell$} (i.north west);
      \draw[cospanarr] (d.south west) to[bend right=10] node[left ]{$a$} (j.north west);
      \draw[cospanarr] (i.south east) to[bend left =10] node[right]{$c$} (k.north east);
      \draw[cospanarr] (j.south east) to[bend right=10] node[below]{$f$} (k.south west);
    \end{tikzpicture}%
    \caption{Structure of a borrowed context diagram.
    The inner, lighter arrows are morphisms of the base category $\catC$,
    while the outer arrows are morphisms of $\ILC(\catC)$.
    }%
    \label{fig-bc-diag}%
  \end{subfigure}\hfill
  \begin{subfigure}[b]{0.52\textwidth}
    \centering
    \begin{tikzpicture}[scale=0.8]
      \def\heightAdjustA{1.00}
      \def\heightAdjustB{1.00}
      \pgfmathsetmacro\posX{2.6}
      \pgfmathsetmacro\posY{-1.85}
      \pgfmathsetmacro\boxSxp{0.9}
      \pgfmathsetmacro\boxSxm{0.65}
      \pgfmathsetmacro\boxSyp{0.56*\heightAdjustB}
      \pgfmathsetmacro\boxSym{0.60*\heightAdjustB}
      \pgfmathsetmacro\boxW{\boxSxp+\boxSxm}
      \pgfmathsetmacro\boxH{\boxSyp+\boxSym}

      \def\circGinner{((0.05,0) ellipse[x radius=0.53, y radius=0.24]}
      \def\circGouter{((0.05,0) ellipse[x radius=0.65, y radius=0.36]}
      \def\circLinner{((0.6,0.14) circle[radius=0.28]}
      \def\circLouter{((0.6,0.14) circle[radius=0.4]}
      \def\allArea{(-\boxSxm,-\boxSym) rectangle (1.1,\boxSyp)} 

      \newcommand{\drawgl}{%
        \draw \circGinner;
        \draw \circGouter;
        \draw \circLinner;
        \draw \circLouter;
      }

      \tikzset{vennfill/.style={fill=black!40}}

      \newcommand{\fillboth}{%
        \fill[vennfill] \circGouter;
        \fill[vennfill] \circLouter;
      }

      \newcommand{\boxat}[5]{
        \begin{scope}[shift={(#1*\posX,#2*\posY*\heightAdjustA)}]
          \node (#3) at ($(\boxSxp,\boxSyp)!.5!(-\boxSxm,-\boxSym)$) [draw,roundbox,cospanintcommon,minimum width=\boxW cm, minimum height=\boxH cm] {};
          \node[anchor=south west] at ($(-\boxSxm,-\boxSym)+(-0.12,-0.15)$) {$\scriptstyle{#4}$};

          \begin{scope}[even odd rule]
            #5
          \end{scope}
          \drawgl
        \end{scope}
      }

      \boxat00DD{
        \clip \circLouter;
        \clip \circGouter;
        \foreach \d in {0,1,...,20} {
          \draw[black!40, line width=0.7pt] (0, -1cm + 0.3pt + \d*2pt) -- +(1cm,1cm);
        }}

      \boxat10LL{
        \fill[vennfill] \circLouter;}

      \boxat20II{
        \clip \circLinner \allArea;
        \fill[vennfill] \circLouter;}

      \boxat01GG{
        \fill[vennfill] \circGouter;}

      \boxat11{Gp}{G^+}{
        \fillboth}

      \boxat21CC{
        \clip \circLinner \allArea;
        \fillboth}

      \boxat02JJ{
        \clip \circGinner \allArea;
        \fill[vennfill] \circGouter;}

      \boxat12FF{
        \clip \circGinner \allArea;
        \fillboth}

      \boxat22KK{
        \clip \circGinner \allArea;
        \clip \circLinner \allArea;
        \fillboth}

      \def\arrowfromto#1#2#3{
        \draw[cospanint] (#1) edge (#2);
      }
      \def\cospanfromto#1#2#3{
        \arrowfromto{#1}{#2}{mono}
        \arrowfromto{#3}{#2}{}
      }
      \def\polabels#1#2#3#4{
        \def\sqlfs{\scriptsize\color{gray}}
        \node at ($(D)!0.5!(Gp)$) {\sqlfs\textup{#1}};
        \node at ($(L)!0.5!(C)$) {\sqlfs\textup{#2}};
        \node at ($(G)!0.5!(F)$) {\sqlfs\textup{#3}};
        \node at ($(Gp)!0.5!(K)$) {\sqlfs\textup{#4}};
      }

      \cospanfromto{D}{L}{I} 
      \cospanfromto{D}{G}{J} 
      \cospanfromto{J}{F}{K}
      \cospanfromto{I}{C}{K}
      \cospanfromto{L}{Gp}{F}
      \cospanfromto{G}{Gp}{C}
      \polabels{JE}{PO}{PO}{PB}
    \end{tikzpicture}
    \caption{Borrowed context diagrams represented as Venn diagrams.
    The outer circles represent graphs $L,G$, and
    the area between the inner and outer circles represents their interfaces $I,J$.}
    \label{fig-bc-venn}
  \end{subfigure}
  \caption{Borrowed context diagrams}
\end{figure}

For cospan categories over adhesive categories, borrowed context
diagrams form a represen\-tative class of
squares~\cite{bchk:conditional-reactive-systems}. Furthermore, for
some categories (such as $\graphfinj$), there are -- up to isomorphism
-- only finitely many jointly epi squares for a given span of monos
and hence only finitely many borrowed context diagrams given $a,\ell$
(since pushout complements along monos in adhesive categories are
unique up to isomorphism).

Whenever the two cospans $\ell,a$ are in $\ILC(\graphfinj)$, it is
easy to see that $f,c$ are in  $\ILC(\graphfinj)$, i.e., they consist
only of monos, i.e., injective morphisms.

Note also that representative squares in $\graphfinj$ are simply
jointly epi squares and they can be straighforwardly extended to
squares of $\ILC(\graphfinj)$.
%
%
%

One central operation is the shift of a condition along an arrow.  The
name shift is taken from an analogous operation for nested application
conditions (see \cite{p:development-correct-gts}).

\begin{definition}[Shift of a Condition]
  \label{def:shift}
  \label{prop:shift}
  Given a fixed class of representative squares $\kappa$,
  the \emph{shift of a condition $\mathcal A$ along an arrow
  $c \colon \RO(\mathcal A) \to B$} is inductively defined as follows:
  \[
    \Big( \bigland_{i \in I} \forall f_i.\mathcal A_i
    \Big)_{\downarrow c} = \bigland_{i \in
        I}\bigland_{(\alpha,\beta) \in \kappa(f_i,c)}
    \forall \beta.({\mathcal A_i}_{\downarrow \alpha})
    \mkern100mu\hfill
    \begin{tikzpicture}[baseline={(0,-0.3)}]
      \path[use as bounding box] (-0.2,0) rectangle (1.6,-0.75);
      \node (sqtl) at (0,0) {};
      \node (sqtr) at (1.25,0) {};
      \node (sqbl) at (0,-0.9) {};
      \node (sqbr) at (1.25,-0.9) {};

      \draw[->] (sqtl) -- node[overlay,above]{$f_i$} (sqtr);
      \draw[->] (sqtl) -- node[left]{$c$} (sqbl);
      \draw[->] (sqtr) -- node[right]{$\alpha$} (sqbr);
      \draw[->] (sqbl) -- node[overlay,below]{$\beta$} (sqbr);
    \end{tikzpicture}
  \]
  Shifting of existential conditions is performed analogously.
\end{definition}


While the
representation of the shifted condition may differ depending on the
chosen class of representative squares, the resulting conditions are
equivalent. Since we assume that each set $\kappa(f_i,c)$ is finite,
shifting a finite condition will again result in a finite condition.
\end{toappendix}

As an example in $\graphfinj$ (the subcategory of $\graphf$ with only injective graph morphisms), shifting $\forall\ \emptyset \to
\fcGraph{n1}{\node[gn] (n1) at (0,0) {$1$}; }
\dotEx
\fcGraph{n1}{\node[gn] (n1) at (0,0) {$1$}; }
\to
\fcGraph{n1}{\node[gn] (n1) at (0,0) {$1$}; \node[gn] (n2) at (1,0) {$2$}; \draw[gedge] (n1) to (n2); }
\dotTrue$
(every node has an outgoing edge)
over
$\emptyset \to \fcGraph{n0}{\node[gn,double] (n0) at (0,0) {\phantom1}; }$
(a node exists)
yields
\begin{align*}
&\forall\ %
\fcGraph{n0}{\node[gn,double] (n0) at (0,0) {\phantom1}; }
\to
\fcGraph{n0}{\node[gn,double] (n0) at (0,0) {\phantom1}; }
\dotEx
\fcGraph{n0}{\node[gn,double] (n0) at (0,0) {\phantom1}; }
\to
\fcGraph{n1}{\node[gn,double] (n0) at (0,0) {}; \node[gn] (n1) at (1,0) {$2$}; \draw[gedge] (n0) to (n1); }
\dotTrue \\
\land\ & \forall\ %
\fcGraph{n0}{\node[gn,double] (n0) at (0,0) {\phantom1}; }
\to
\fcGraph{n1}{\node[gn,double] (n0) at (-0.75,0) {}; \node[gn] (n1) at (0,0) {$1$}; }
\ .\big(
\exists\ %
\fcGraph{n1}{\node[gn,double] (n0) at (-0.75,0) {}; \node[gn] (n1) at (0,0) {$1$}; }
\to
\fcGraph{n1}{\node[gn,double] (n0) at (-0.75,0) {}; \node[gn] (n1) at (0,0) {$1$}; \node[gn] (n2) at (1,0) {$2$}; \draw[gedge] (n1) to (n2); }
\dotTrue \lor \exists\ %
\fcGraph{n1}{\node[gn,double] (n0) at (-0.75,0) {}; \node[gn] (n1) at (0,0) {$1$}; }
\to
\fcGraph{n1}{\node[gn,double] (n0) at (-1,0) {}; \node[gn] (n1) at (0,0) {$1$}; \draw[gedge] (n1) to (n0); }
\dotTrue \big)
\end{align*}
(the designated node has an outgoing edge, and so does every other
node, possibly to the designated node).  This example can be lifted to
$\ILC(\graphf)$ by replacing all graph morphisms $A \xrightarrow{m} B$
with cospans $A \xrightarrow{m} B \xleftarrow{\id} B$.

\begin{toappendix}
  \paragraph*{Visualization of shifts}
  Given a condition $\mathcal{A}$ and an arrow $c\colon A =
  \RO(\mathcal{A})\to B$, we will visualize shifts in diagrams as follows:

\noindent\mbox{}\hfill\begin{tikzpicture}
  \def\sqw{1.75}
  \node (a) at (0,0) {$A$}; \node (b) at (1*\sqw,0) {$B$}; \node (x)
  at (2*\sqw,0) {$X$};
  \draw[->] (a) -- node[above]{$c$} (b); \draw[->] (b) --
  node[above]{$d$} (x); \node[condtri,dart tip angle=30,shape border
  rotate=270,rotate around={-10:(a.center)},
  label={[rotate=0,anchor=south,label
    distance=1pt]above:{$\mkern7mu\mathcal
      A\vphantom{{}_{\downarrow}}$}}] at (a.north) {\kern3pt};
  \node[condtri,dart tip angle=30,shape border rotate=270,rotate
  around={-10:(b.center)}, label={[rotate=0,anchor=south,label
    distance=1pt]above:{$\mathcal A_{\downarrow c}$}}] at (b.north)
  {\kern3pt};
\end{tikzpicture}\hfill\mbox{}

Remember that for an arrow $d\colon B\to X$ it holds that $d\models
\mathcal{A}_{\downarrow c} \iff c;d \models \mathcal{A}$.
\end{toappendix}

In order to compute successor states for graph conditions, we need the
concepts of Hoare triple, (strongest) postconditions and (weakest)
preconditions that is based on the shift operation.

\begin{definition}[Hoare triple, weakest precondition, strongest postcondition~\cite{bchk:conditional-reactive-systems}]
  Let $\mathcal R = (\ell,r,\mathcal C)$ be a rule and let $\mathcal A, \mathcal B$ be conditions.
  We say that $\mathcal A, \mathcal R, \mathcal B$ form a \emph{Hoare
    triple} -- written as $\{ \mathcal A \} \mathcal R \{ \mathcal B
  \}$ -- if for all $a,b \colon \obnull \to J$ with $a \models \mathcal A$ and $a \leadsto_{\mathcal R} b$ we have that $b \models \mathcal B$.

  $\mathcal A$ is a \emph{precondition} for $\mathcal R$ and $\mathcal B$ whenever $\{\mathcal A\} \mathcal R \{\mathcal B\}$.
  Similarly, $\mathcal B$ is called a \emph{postcondition} for $\mathcal A$ and $\mathcal R$.

  $\mathcal A$ is the \emph{weakest precondition} for $\mathcal R$ and $\mathcal B$ (written $\WP(\mathcal R, \mathcal B)$) whenever it is a precondition and for every other precondition $\mathcal A'$ we have that $\mathcal A' \models \mathcal A$.

  $\mathcal B$ is the \emph{strongest postcondition} for $\mathcal A$ and $\mathcal R$ (written $\SP(\mathcal A, \mathcal R)$) whenever it is a postcondition and for every other postcondition $\mathcal B'$ we have that $\mathcal B \models \mathcal B'$.
\end{definition}

It is easy to see that $\{\mathcal A\} \mathcal R \{\mathcal B\}$ iff
$\mathcal{A}\models \WP(\mathcal R, \mathcal B)$ iff
$\SP(\mathcal A, \mathcal R)\models \mathcal{B}$. Furthermore all
notions can be generalized to traces, i.e., sequences of rules,
instead of single rules $\cal R$. 

\begin{proposition}[Computing weakest preconditions and strongest postconditions~\cite{bchk:conditional-reactive-systems}]
  \label{def:compute-sp-wp}
  Let $\mathcal R=(\ell,r,\mathcal C)$ be a rule and let $\mathcal A, \mathcal B$ be conditions.
  Then
  $\WP(\mathcal R, \mathcal B) \equiv \forall \ell . (\mathcal C \rightarrow \mathcal B_{\downarrow r})$
  and
  $\SP(\mathcal A, \mathcal R) \equiv \exists r . (\mathcal C \land \mathcal A_{\downarrow \ell})$
\end{proposition}

For instance, for the motivating example in \Cref{sec:motivation}, the strongest postcondition of the first step is the following.
%
\begin{align*}
  & \SP(\Init_1 \land \neg\Bad, \append) =
  \exists\ \emptyset \to
    \fcGraph{n1}{
      \node[gn] (n1) at (0,0) {$1$};
      \node[gn] (n2) at (1,0) {$2$};
      \draw[gedge] (n1) to (n2)
                   (n2) edge[exaloop right] (n2);
    }
    \leftarrow \fcGraph{n1}{ \node[gn] (n1) at (0,0) {$1$}; }
    \ . \big( \\
    & \quad
      \phantom{\land\ }
      \forall\ \fcGraph{n1}{ \node[gn] (n1) at (0,0) {$1$}; }
        \rightarrow \fcGraph{n1}{
          \node[gn] (n1) at (0,0) {$1$};
          \node[gn] (n3) at (1,0) {$3$};
          \draw[gedge] (n1) to (n3);
        } \leftarrow \fcGraph{n1}{ \node[gn] (n1) at (0,0) {$1$}; \node[gn] (n3) at (0.8,0) {$3$}; }
      \dotFalse \mkern47.3mu
      \land \forall\ \fcGraph{n1}{ \node[gn] (n1) at (0,0) {$1$}; }
        \rightarrow \fcGraph{n1}{
          \node[gn] (n1) at (0,0) {$1$};
          \node[gn] (n3) at (1,0) {$3$};
          \draw[gedge] (n1) edge[exaloop above] (n1)
                       (n1) to (n3);
        } \leftarrow \fcGraph{n1}{ \node[gn] (n1) at (0,0) {$1$}; \node[gn] (n3) at (0.8,0) {$3$}; }
      \dotFalse \\
    & \quad
      \land \forall\ \fcGraph{n1}{ \node[gn] (n1) at (0,0) {$1$}; }
        \rightarrow \fcGraph{n1}{
          \node[gn] (n0) at (0,0) {$0$};
          \node[gn] (n1) at (1,0) {$1$};
          \draw[gedge] (n0) to (n1);
        } \leftarrow \fcGraph{n1}{ \node[gn] (n0) at (0,0) {$0$}; \node[gn] (n1) at (0.8,0) {$1$}; }
      \dotFalse \mkern47.3mu
      \land \forall\ \fcGraph{n1}{ \node[gn] (n1) at (0,0) {$1$}; }
        \rightarrow \fcGraph{n1}{
          \node[gn] (n0) at (0,0) {$0$};
          \node[gn] (n1) at (1,0) {$1$};
          \draw[gedge] (n0) edge[exaloop above] (n0)
                       (n0) to (n1);
        } \leftarrow \fcGraph{n1}{ \node[gn] (n0) at (0,0) {$0$}; \node[gn] (n1) at (0.8,0) {$1$}; }
      \dotFalse \\
    & \quad
      \land \forall\ \fcGraph{n1}{ \node[gn] (n1) at (0,0) {$1$}; }
        \rightarrow \fcGraph{n1}{
          \node[gn] (n1) at (0,0) {$1$};
          \node[gn] (n3) at (0.8,0) {$3$};
          \node[gn] (n4) at (1.8,0) {$4$};
          \draw[gedge] (n3) to (n4);
        } \leftarrow \fcGraph{n1}{ \node[gn] (n1) at (0,0) {$1$}; \node[gn] (n3) at (0.8,0) {$3$}; \node[gn] (n4) at (1.6,0) {$4$}; }
      \dotFalse
      \land \forall\ \fcGraph{n1}{ \node[gn] (n1) at (0,0) {$1$}; }
        \rightarrow \fcGraph{n1}{
          \node[gn] (n1) at (0,0) {$1$};
          \node[gn] (n3) at (0.8,0) {$3$};
          \node[gn] (n4) at (1.8,0) {$4$};
          \draw[gedge] (n3) edge[exaloop above] (n3)
                       (n3) to (n4);
        } \leftarrow \fcGraph{n1}{ \node[gn] (n1) at (0,0) {$1$}; \node[gn] (n3) at (0.8,0) {$3$}; \node[gn] (n4) at (1.6,0) {$4$}; }
      \dotFalse
    \big)
\end{align*}
Essentially, this states that a list with one element must exist, and condition $\Init_1 \land \neg\Bad$ has to hold for both the second-to-last list element that has just been added, and any other list elements that might exist.
Note that the three subconditions in the right column are already ``covered'' by the three ones on the left, and could be removed to obtain a smaller but equivalent condition.

\section{GTS verification using predicate abstraction}

We are now in a position to formalize the method outlined and illustrated in \Cref{sec:motivation}. To reiterate: we want to answer verification questions of the form ``from a given initial system state, is it possible to reach a state where a given (undesirable) property holds?'' --- where, for us, states are elements of $\Arr_0$ in an arbitrary base category $\catC$ with distinguished object $0$.

\subsection{Concrete transition systems}

As a first observation, in practice we are interested in answering the verification question for a \emph{family} of initial system states, and not just a single one. We therefore immediately generalise the formal notion of states and transitions to sets of arrows, with a disjunctive interpretation: a system being ``in'' a state means that it is described by \emph{one} of the elements of that state.

\begin{definition}[Set-based transition system, $\RULE(\HOLE\mkern-0.5mu,\mkern-1mu\mathcal R\mkern-1.5mu)$, correctness]
  Given a reactive system $\mathcal S$, a \emph{set-based transition
  system} is a tuple $T=\tup{Q,{\rightarrow},X_0}$, where $Q \subseteq \PArr{0}$ is the set of states, $X_0\in Q$ is the initial state, and ${\rightarrow}\subseteq Q\times \mathcal S\times Q$ is the 
  transition relation, defined by $X\xrightarrow{\mathcal R} \RULE(X,\mathcal R)$ where $\RULE(X,\mathcal R) \defeq \{ y \mid x \in X,\ x \leadsto_{\mathcal R} y \}$ for arbitrary $X\in Q$ and $\mathcal R\in \mathcal S$.
  
  $T$ is called \emph{correct} with respect to a given condition $\Bad\in\Cond_0$ if $Y\cap \sem\Bad=\emptyset$ for all states $Y$ reachable from $X_0$.
\end{definition}

For our running example we obtain the following set-based transition
system, starting from the set of arrows that satisfy $\Init_2$ (rule $\mathcal R = \append$):
\[
  \{
    \emptyset,\ %
    \fcGraph{n1}{\node[gn] (n1) at (0,0) {$\phantom1$}; },\ %
    \fcGraph{n1}{\node[gn] (n1) at (0,0) {$\phantom1$}; \draw[gedge] (n1) edge[exaloop above] (n1); },\ %
    \fcGraph{n1}{\node[gn] (n1) at (0,0) {$\phantom1$}; \node[gn] (n2) at (0.75,0) {$\phantom1$}; },\ %
    \fcGraph{n1}{\node[gn] (n1) at (0,0) {$\phantom1$}; \node[gn] (n2) at (0.75,0) {$\phantom1$}; \draw[gedge] (n2) edge[exaloop above] (n2); },\ %
    \fcGraph{n1}{\node[gn] (n1) at (0,0) {$\phantom1$}; \node[gn] (n2) at (0.75,0) {$\phantom1$}; \draw[gedge] (n1) edge[exaloop above] (n1); \draw[gedge] (n2) edge[exaloop above] (n2); },\ %
    \tdots \}
  \xrightarrow{\mathcal R}
  \{\ %
    \fcGraph{n1}{\node[gn] (n1) at (0,0) {$\phantom1$}; \node[gn] (n2) at (1,0) {$\phantom1$}; \draw[gedge] (n2) edge[exaloop above] (n2); \draw[gedge] (n1) edge (n2); },\ %
    \fcGraph{n1}{\node[gn] (n1) at (0,0) {$\phantom1$}; \node[gn] (n2) at (1,0) {$\phantom1$}; \draw[gedge] (n2) edge[exaloop above] (n2); \draw[gedge] (n1) edge (n2); \node[gn] (n3) at (1.75,0) {$\phantom1$}; },\ %
    \fcGraph{n1}{\node[gn] (n1) at (0,0) {$\phantom1$}; \node[gn] (n2) at (1,0) {$\phantom1$}; \draw[gedge] (n2) edge[exaloop above] (n2); \draw[gedge] (n1) edge (n2); \node[gn] (n3) at (1.75,0) {$\phantom1$}; \draw[gedge] (n3) edge[exaloop above] (n3); },\ %
    \tdots \}
  \xrightarrow{\mathcal R}
  \cdots
\]

A set-based transition system is \emph{induced} by some condition
$\Init$ if $X_0=\sem\Init$ and $Q$ is the smallest subset of $\PArr 0$
reachable from $X_0$. For an induced set-based transition system, the
verification question therefore asks whether any
reachable set $Y$ intersects with $\sem{\Bad}$ for a condition $\Bad$.

Since individual states as well as the set of all states can be
infinite, verification on set-based transition systems is in general
infeasible. Note, however, that the problem of checking whether a
system is \emph{in}correct (a bad state is reachable), is in fact
semi-decidable for graphs (rewriting in $\ILC(\graphf)$): we can
enumerate all graphs satisfying $\Init$ and while doing this in
parallel enumerate the reachable graphs. Once we detect a graph
satisfying $\Bad$, we can give the respective answer, i.e., the system
is not correct. However, it is well-known that graph transformation
systems can encode Turing machines~\cite{hp:languages-gts} and hence
the problem is undecidable. Here we are interested in developing a
technique that can (in some cases) definitely show that the system
under consideration is in fact correct.

A first step towards a verification method is to use conditions
(which have finite size) as a representation for (first-order definable) sets of arrows. This step is justified by the following result, which implies that the set of condition-definable sets of arrows is closed under rule application.

\begin{lemmarep}
  \label{lemma-rule-of-a-is-fo-definable}
  For any condition $\cal A$ and rule $\mathcal R$,
  $\sem{\SP(\mathcal A, \mathcal R)} =
  \RULE(\sem{\mathcal A}, \mathcal R)$. 
\end{lemmarep}

\begin{proof}
  \begin{align*}
    g \in \sem{\SP(\mathcal A, (\ell,r,\mathcal C))}
    &\iff g \models \SP(\mathcal A, (\ell,r,\mathcal C)) \\
    &\iff g \models \exists r . (\mathcal C \land \mathcal A_{\downarrow \ell}) \\
    \text{(Def. $\models$)} &\iff \exists c \colon g = r;c \land c \models \mathcal C \land c \models \mathcal A_{\downarrow \ell} \\
    \text{(Def. Shift)} &\iff \exists c \colon g = r;c \land c \models \mathcal C \land \ell;c \models \mathcal A \\
    &\iff \exists f \colon f \leadsto_{(\ell,r,\mathcal C)} g \land f \models \mathcal A \\
    &\iff g \in \RULE(\sem{\mathcal A}, (\ell,r,\mathcal C))
  \end{align*}
  As a consequence, $\RULE(\sem{\mathcal A}, \mathcal R)$ is
  definable by a condition.
  \qed
\end{proof}

Using the construction given in \Cref{sec:shift-hoare}, we can define transitions through strongest postconditions. However, we have to ensure that equivalent but syntactically distinct conditions collapse to the same state. This gives rise to the following definition:

\begin{definition}[Condition-based transition system]
  \label{def:trans-conditions}
  Given a reactive system $\mathcal S$, a \emph{condition-based transition system} is a tuple $\tup{Q,{\rightarrow},\Init}$ with $\Init\in \Cond_0$, where $Q\subseteq  \Cond_0/{\equiv}$ is the set of states, $[\Init]_\equiv\in Q$ is the initial state, and ${\rightarrow}\subseteq Q\times \mathcal S\times Q$ is the transition relation, defined by $[\mathcal A]_\equiv \xrightarrow{\mathcal R} [\SP(\mathcal A,\mathcal
  R)]_\equiv$ for arbitrary $\mathcal A\in \Cond_0$ and $\mathcal R\in \mathcal S$.
\end{definition}

Transitions are well-defined because
$\mathcal A\equiv \mathcal B$ implies
$\SP(\mathcal A,\mathcal R)\equiv \SP(\mathcal B,\mathcal R)$.
Below we will usually omit the explicit construction of
$\equiv$-equivalence classes and just talk about conditions, tacitly
assuming that they are representatives of the corresponding equivalence
classes. Due to \Cref{lemma-rule-of-a-is-fo-definable},
$\sem\_$ maps any condition-based transition system to an isomorphic
set-based one, with initial state $\sem{\Init}$.

On condition-based transition systems,
the verification problem (is a transition system correct w.r.t.\ $\Bad$) reduces to checking whether all states reachable from $\mathcal A$ entail $\lnot\Bad$.
The condition-based transition system for our running example has the
following initial steps, starting with $\Init_2$ (cf.\ the set-based transition system above):
\[
  \Init_2
  \xrightarrow{\smash{\mathcal R}}
  \SP(\Init_2, \mathcal R)
  \xrightarrow{\smash{\mathcal R}}
  \SP(\SP(\Init_2, \mathcal R), \mathcal R)
  \xrightarrow{\smash{\mathcal R}}
  \cdots
\]
The condition in the initial state expresses that there exists a multiset of empty lists.
The second state (after a single rule application) allows a single list element in any of the lists.
The third state (after two rule applications) allows two list elements in total, either as two one-element lists or as a single two-element list, and otherwise only empty lists. None of these conditions are equivalent.

\subsection{Abstract transition systems}

Condition-based transition systems still do not provide a way to
answer the verification question: compared to the set-based transition
systems, successors are now representable; however, the definition of
transitions relies on entailment, which is undecidable, and the
reachable part of the transition system will typically still be
infinite and therefore not fully explorable (as in the example
above). This is where we introduce predicate abstraction.  Instead of
conditions of arbitrary complexity, states will be subsets (or
conjunctions) of a predetermined set of conditions (the
\emph{predicates}), each of which can be either positive, negative or
absent (unknown) (e.g.\
$(\mathcal P_1 \land \neg \mathcal P_3) \in \Abs(\{\mathcal
P_1,\mathcal P_2,\mathcal P_3\})$).
This guarantees finiteness of the resulting transition system.

\begin{definition}[Predicate abstraction]
  Let $P = \{ \mathcal{P}_1, \dots, \mathcal{P}_n \}$ be a non-empty set of
  conditions in $\Cond_0$ which we will call \emph{predicates}.  We
  define a lattice $\Abs(P)$ as follows:
  \begin{itemize}
  \item The carrier set contains all conjunctions of subsets of
    $P \cup \{\neg \mathcal{P}_1, \tdots,
    \neg \mathcal{P}_n\}$, quotiented by equivalence $\equiv$ (which includes the constants $\condtrue$ and
        $\condfalse$).
  \item The set is ordered by entailment (\kern2pt$\models$).
  \end{itemize}

  \noindent
  For an arbitrary condition $\mathcal{A} \in \Cond_0$,
  $\overline{\mathcal{A}} \defeq \bigland \{ \mathcal{Q}' \in \Abs(P) \mid
  \mathcal{A} \models \mathcal{Q}' \}$ is the strongest element of
  $\Abs(P)$ for which $\mathcal{A} \models \overline{\mathcal{A}}$, i.e.,
  the best possible approximation of $\mathcal{A}$ 
  for the given set of predicates.
\end{definition}

Since $\overline{\mathcal A}$ is in general weaker than $\mathcal A$, reasoning with $\overline{\mathcal A}$ rather than $\mathcal A$ results in over-approximation, meaning that our
abstract transition system suggests that the reachable sets of arrows
are larger than is actually the case. As a result, unsafe states might
seemingly be reachable when in reality they are not. Avoiding this
requires careful selection of a suitable set of predicates. We will
take care of this issue later in \Cref{sec:cegar}.

\begin{definition}[Abstract transition system, $\SP^\#(\HOLE,\mathcal R)$]
  \label{def:ts-for-predicates}
  Given a reactive system $\mathcal S$ and a set of predicates $P$, an \emph{abstract transition system} is a tuple $\tup{Q,{\Rightarrow},\Init}$ with $\Init\in P$, where $Q\subseteq \Abs(P)$ is the set of states, $[\Init]_\equiv \in Q$ is the initial state, and ${\Rightarrow}\subseteq Q\times \mathcal S\times Q$ is the transition relation, defined by $\mathcal Q\abstractTransition{\mathcal R} \SP^\#(\mathcal Q,\mathcal R)$ where $\SP^\#(\mathcal{Q},\mathcal R)\defeq\overline{\SP(\mathcal Q,\mathcal R)}$.
\end{definition}

Hence the abstract transition relation is obtained by
computing the strongest postcondition of a condition and then
weakening it so that it can be expressed in $\Abs(P)$. The latter
requires checking whether $\SP(\mathcal{Q}, \mathcal R)\models \mathcal{P}_i$ or
$\SP(\mathcal{Q}, \mathcal R)\models \lnot \mathcal{P}_i$ for all $i$ and forming a
conjunction of those predicates where the check succeeds.

In fact, this approach precisely follows the paradigm of abstract interpretation, based on Galois connections. Let $\alpha$ and $\gamma$ be mappings from $\Cond_0/{\equiv}$ to $\Abs(P)$ and back, defined by $\alpha(\mathcal A) \defeq \overline{\mathcal A}$ and $\gamma(\mathcal Q)\defeq\mathcal Q$, respectively. We then have the following:

\begin{propositionrep}
  \label{prop:a-g-is-galois}
  Let $P$ be a set of predicates. Then $(\alpha,\gamma)$ as defined
  above is a Galois connection between $\Cond/{\equiv}$ and $\Abs(P)$,
  and $\SP^\#$ is the induced over-approximation of $\SP$ (i.e.,
  $\SP^\#(\HOLE,\mathcal{R})=\alpha\circ \SP(\HOLE,\mathcal{R})\circ
  \gamma$).
\end{propositionrep}

\begin{proof}
  We first check that $(\alpha,\gamma)$ form a Galois
  connection. Given $\mathcal{A}\in\Cond_0/{\equiv}$, we have that
  \[ \gamma(\alpha(\mathcal{A})) = \alpha(\mathcal{A}) =
    \overline{\mathcal{A}} = \bigwedge \{\mathcal{Q}'\in \Abs(P)\mid
    \mathcal{A}\models \mathcal{Q}'\} \sledom \mathcal{A}. \]
  For the other inequality assume that $\mathcal{Q}\in \Abs(P)$ and we
  obtain
  \[ \alpha(\gamma(\mathcal{Q})) = \alpha(\mathcal{Q}) =
    \overline{\mathcal{Q}} = \bigwedge \{\mathcal{Q}'\in \Abs(P)\mid
    \mathcal{Q}\models \mathcal{Q}'\} \equiv \mathcal{Q}. \]
  The equivalence holds since $\mathcal{Q}$ itself is in $\Abs(P)$ and
  is entailed by $\mathcal{Q}$.

  Finally we observe that for $\mathcal{Q}\in\Abs(P)$, we have
  \[ \alpha(\SP(\gamma(\mathcal{Q}),\mathcal{R})) =
    \overline{\SP(\mathcal{Q},\mathcal{R})} =
    \SP^\#(\mathcal{Q},\mathcal{R}). \] Hence
  $\SP^\#(\HOLE,\mathcal{R})=\alpha\circ \SP(\HOLE,\mathcal{R})\circ
  \gamma$. \qed
\end{proof}

Note that
$\alpha$ is (in general) not computable because it
involves the entailment problem of first-order logic. A practical
solution to its non-computability will be discussed later in the paper
in \Cref{sec:undecidability-issues}.

For our running example, we have used this construction in \Cref{sec:motivation}, first for $P = \{\Init_2, \Bad\}$ and next for $P = \{\Init_2, \Bad, \mathcal W_1\}$, to  obtain the abstract transition systems induced by $\Init_2$.
In the second case, as all states entail $\neg \Bad$, we
verified the desired property.

\begin{toappendix}
  \begin{lemma}
    \label{lem:concrete-implies-abs}
    Let $\mathcal R_1, \dots, \mathcal R_n$ be a rule sequence.

    Let $\mathcal{A}_i$ be the conditions of the corresponding run in
    the condition-based transition system, starting from the initial
    condition, i.e., $\mathcal{A}_0 = \Init \in \Abs(P)$ and
    $\mathcal{A}_{i+1} = \SP(\mathcal{A}_i,\mathcal{R}_{i+1})$
    (i.e. $\mathcal{A}_i\xrightarrow{\mathcal{R}_{i+1}} \mathcal{A}_{i+1}$).

    Let $\mathcal Q_i$ be the conditions of the corresponding abstract run, i.e.,
    $\mathcal Q_0 = \mathcal{A}_0 = \Init$ and
    $\mathcal{Q}_{i+1} = \SP^\#(\mathcal{Q}_i,\mathcal{R}_{i+1})$
    (i.e., $\mathcal{Q}_i \abstractTransition{\mathcal{R}_{i+1}} \mathcal{Q}_{i+1}$).

    Then, $\mathcal{A}_i \models \mathcal Q_i$ for all $i$.
  \end{lemma}

  \begin{proof}
    Using \Cref{prop:a-g-is-galois}, we first observe that for any
    $\mathcal{A} \in \Abs(P)$ we have
    $\SP(\mathcal{A},\mathcal R) \models \SP^\#(\mathcal{A}, \mathcal
    R)$:
      \[ \SP(\mathcal{A},\mathcal{R}) \models
      \overline{\SP(\mathcal{A},\mathcal{R})} =
      \alpha(\SP(\gamma(\mathcal{A}),\mathcal{R})) =
      \SP^\#(\mathcal{A},\mathcal{R}) \]
    Now we show $\mathcal{A}_i \models \mathcal Q_i$ by induction.
    \begin{itemize}
      \item $i=0$: trivial
      \item $i \to i+1$:
        Given $\mathcal{A}_i \models \mathcal Q_i$,
        we have $\SP(\mathcal{A}_i, \mathcal R_{i+1}) \models \SP(\mathcal Q_i, \mathcal R_{i+1})$ since $\SP$ is monotone.

        As $\SP(\mathcal{A},\mathcal R) \models \SP^\#(\mathcal{A}, \mathcal R)$ (shown above),
        also $\SP(\mathcal Q_i, \mathcal R_{i+1}) \models \SP^\#(\mathcal Q_i, \mathcal R_{i+1})$.

        In total: $\mathcal{A}_{i+1} = \SP(\mathcal{A}_i, \mathcal R_{i+1}) \models \SP^\#(\mathcal Q_i, \mathcal R_{i+1}) = \mathcal Q_{i+1}$.
        \qed
    \end{itemize}
  \end{proof}
\end{toappendix}

\begin{theoremrep}
  \label{thm:abstract-ts-correct}
  Let $P$ be a set of predicates with $\Init,\Bad\in P$. If all reachable states of the abstract transition system with initial state $\Init$ entail $\lnot \Bad$, the set-based transition system induced by $\Init$ is correct w.r.t.\ $\Bad$.
\end{theoremrep}

\begin{proof}
  Assume by contradiction
  that the system is not correct, that is, there exists a rule
  sequence $\mathcal{R}_1,\dots,\mathcal{R}_n$ such that we have the
  following transitions in the (concrete) set-based transition
  system
  \[ \llbracket\Init\rrbracket \xrightarrow{\mathcal{R}_1} X_1
    \xrightarrow{\mathcal{R}_2} \dots \xrightarrow{\mathcal{R}_{n-1}}
    X_{n-1} \xrightarrow{\mathcal{R}_n} X_n \] where
  $X_n\cap\llbracket \Bad\rrbracket \neq \emptyset$. Note that here
  $X_{i+1} = \RULE(X_i,\mathcal{R}_{i+1})$.

  Let $\mathcal A_i$ be the conditions of the corresponding run in the
  (concrete) condition-based transition system: we define
  $\mathcal{A}_0 = \Init$,
  $\mathcal{A}_{i+1} = \SP(\mathcal{A}_i,\mathcal{R}_{i+1})$. By
  induction, using \Cref{lemma-rule-of-a-is-fo-definable}, we obtain
  $\llbracket \mathcal{A}_i\rrbracket = X_i$.

  Now let $\mathcal Q_i$ be the conditions of the corresponding
  abstract run: define $\mathcal{Q}_0 = \Init$,
  $\mathcal{Q}_{i+1} = \SP^\#(\mathcal{Q}_i,\mathcal{R}_{i+1})$. In
  particular,
  $\mathcal{Q}_i \abstractTransition{\mathcal{R}_{i+1}}
  \mathcal{Q}_{i+1}$.

  By \Cref{lem:concrete-implies-abs} we have
  $\mathcal{A}_i\models \mathcal{Q}_i$ for all $i$. This implies
  $\llbracket \mathcal{A}_i\rrbracket \subseteq \llbracket
  \mathcal{Q}_i\rrbracket$.

  Since
  $\llbracket \mathcal{A}_n\rrbracket \cap\llbracket \Bad\rrbracket =
  X_n\cap \llbracket \Bad\rrbracket \neq \emptyset$, we obtain that
  $\llbracket \mathcal{Q}_n\rrbracket \cap\llbracket \Bad\rrbracket
  \neq \emptyset$. This implies that
  $\llbracket \mathcal{Q}_n\rrbracket \not\subseteq
  \Cond_0\backslash\llbracket \Bad\rrbracket = \llbracket \lnot\Bad
  \rrbracket$, which implies $\mathcal{Q}_n\notmodels \lnot
  \Bad$. But this is a contradiction since $\mathcal{Q}_n$ is a state
  reachable in the abstract transition system that implies
  $\lnot \Bad$ by assumption.
  %
  %
  %
  %
  \qed
\end{proof}

\begin{toappendix}
\subsection{Abstraction and concretization via Galois
  connections}

In the theory of abstract interpretation
\cite{c:abstract-interpretation,cc:ai-unified-lattice-model} one
usually employs a Galois connection to connect the concrete and the
abstract domain (cf.\ \Cref{sec:abstract-interpretation}),
allowing to give a uniform treatment.

We have already seen one such Galois connection in this paper:
$(\alpha, \gamma)$, furthermore the concretization map $\llbracket\_\rrbracket$.
In the diagram below we give a more systematic overview over the
various abstraction and concretization maps used in the paper and
their properties.

\noindent\mbox{}\hfill\begin{tikzpicture}[x=3cm,y=-1cm]
  \node (pgraph) at (0,0) {$\PArr{0}$};
  \node (cond) at (1,1) {$\Cond_0/\!\equiv$};
  \node (abspi) at (2,0) {$\Abs(P)$};
  \draw[->] (cond) to[bend left=7] node[below]{$\llbracket\_\rrbracket$} (pgraph);
  \draw[->] (cond) to[bend left=7] node[above,pos=0.25]{$\alpha\ $} (abspi);
  \draw[->] (abspi) to[bend left=7] node[below]{$\gamma$} (cond);
  \draw[->] (pgraph) to[bend left=15] node[above]{$\alpha_A$} (abspi);
  \draw[->] (abspi.177) to[bend right=8] node[below]{$\gamma_A$} (pgraph.3);
  \draw[->] (pgraph) to[out=170,in=-170,looseness=4] node[left]{$\RULE(\HOLE, \mathcal R)$} (pgraph);
  \draw[->] (abspi) to[out=12,in=-12,looseness=5] node[right]{$\SP^\#(\HOLE, \mathcal R)$} (abspi);
  \draw[->] (cond) to[out=-60,in=-120,looseness=5] node[below]{$\SP(\HOLE, \mathcal R)$} (cond);
\end{tikzpicture}\hfill\mbox{}

First note that we cannot define a Galois connection between
$\PArr{0}$ and $\Cond_0$ because the corresponding abstraction (left
adjoint to $\llbracket\_\rrbracket$) cannot be defined: For any
non-first-order-definable set of graphs, there is a series of graph
conditions, providing successively better over-approximations of the
set as the conditions increase in size, but there is only a unique
best over-approximation if we restrict to first-order-definable sets
of graphs.  



However, somewhat surprisingly, there exists a Galois connection
between $\PArr{0}$ and $\Abs(P)$ (with $P=\{\mathcal{P}_1,\dots,\mathcal{P}_n\}$) that can
be defined as follows:
\begin{eqnarray*}
  \alpha_A(X) & \defeq & \bigwedge \{ \mathcal{Q} \mid
  \mathcal{Q}\in \{ \mathcal{P}_1,\neg
  \mathcal{P}_1,\dots,\mathcal{P}_n,\neg \mathcal{P}_n
  \}, \forall x \in X \colon x \models \mathcal{Q} \} \\
  \gamma_A(\mathcal{Q}) & \defeq & \{ x\in \Arr_0 \mid x\models \mathcal{Q} \}
\end{eqnarray*}

It is easy to see that it is a Galois connection.

\begin{lemmarep}
  \label{lem:comm-galois}
  It holds that $\gamma_A = \llbracket\_\rrbracket \circ \gamma$ and
  $\alpha = \alpha_A\circ\llbracket\_\rrbracket$.
\end{lemmarep}

\begin{proof}
  Given $\mathcal{Q}\in\Abs(P)$, we have:
  \[
    \llbracket\gamma(\mathcal{Q})\rrbracket = \llbracket
    \mathcal{Q}\rrbracket = \{x\in \Arr_0\mid x\models \mathcal{Q}\} =
    \gamma_A(\mathcal{Q}) \]
  Furthermore, given $\mathcal{Q}\in \Cond_0/\!\equiv$, we obtain:
  \begin{eqnarray*}
    \alpha_A(\llbracket \mathcal{Q}\rrbracket) & = & \bigwedge \{ \mathcal{Q} \mid
    \mathcal{Q}'\in \{ \mathcal{P}_1,\neg
    \mathcal{P}_1,\dots,\mathcal{P}_n,\neg \mathcal{P}_n
    \}, \forall x \in \llbracket Q\rrbracket \colon x \models
    \mathcal{Q}' \} \\
    & = &
    \bigland \{ \mathcal{Q}' \in \Abs(P) \mid
    \mathcal{Q} \models \mathcal{Q}' \} \\
    & = & \overline{\mathcal A} = \alpha(\mathcal A)
  \end{eqnarray*}
  \vskip-\baselineskip\qed
\end{proof}

The following lemma shows that the $\SP^\#$ is also the
over-approximation that is induced by $\HOLE$ and the Galois
connection $(\alpha_A,\gamma_A)$, meaning that we could have based our
developments on it instead of $(\alpha,\gamma)$.

\begin{lemmarep}
  \label{lem:sphash-induced-approx}
  $\SP^\#$ is the the induced over-approximation of
  $\RULE(\HOLE, \mathcal R)$ via the Galois connection
  $(\alpha_A,\gamma_A)$.
\end{lemmarep}

\begin{proof}
  We have to show that
  $\alpha_A(\RULE(\gamma_A(\mathcal{P}), \mathcal R)) =
  \SP^\#(\mathcal{P}, \mathcal R)$. And indeed we have:
  \begin{align*}
    & \alpha_A(\RULE(\gamma_A(\mathcal{P}), \mathcal R)) \\
    \text{(\Cref{lem:comm-galois})} &= \alpha_A(\RULE(\llbracket \gamma(\mathcal{P})\rrbracket, \mathcal R)) \\
    \text{(\Cref{lemma-rule-of-a-is-fo-definable})}
    &= \alpha_A(\llbracket\SP(\gamma(\mathcal{P}), \mathcal R)\rrbracket) \\
    \text{(\Cref{lem:comm-galois})} &= \alpha(\SP(\gamma(\mathcal{P}), \mathcal R)) =
    \SP^\#(\mathcal{P}, \mathcal R)
  \end{align*}
  \vskip-\baselineskip\qed
\end{proof}
\end{toappendix}

\section{Counterexample-guided abstraction refinement (CEGAR)}
\label{sec:cegar}

We are now ready to define the full CEGAR loop. In particular, we will
explain how to obtain suitable predicates for refinement.

\subsection{Obtaining predicates}
\label{sec:obtaining-predicates}

In the example from \Cref{sec:motivation}, using only predicates $P=\{\Init_2,\Bad\}$ and initial state $\Init_2$ we found an apparently unsafe abstract state, i.e., one which did not entail $\neg\Bad$, through the trace $\append,\append$. However, the condition-based state reached via the same trace was actually safe, and augmenting $P$ with $\mathcal W_1$ resulted in a successful proof that, indeed, all reachable states are safe. Hence the general question arises how to refine a set of predicates, given an abstract trace to an unsafe state (i.e., a \emph{counterexample} to correctness) that does not exist on the concrete level (i.e., is \emph{spurious}).

\begin{definition}[Spurious counterexample]
  \label{def:spurious-cex}
  Let $\mathcal S$ be a reactive system and $P$ be a set of predicates with $\Init,\Bad \in P$, and consider the abstract transition system with initial state $\Init$. A \emph{counterexample} to correctness w.r.t.\ $\Bad$ is a trace $\mathcal R_1\cdots \mathcal R_n\in \mathcal S^*$  such that $\Init \abstractTransition{\mathcal R_1} \mathcal{Q}_1 \abstractTransition{\mathcal R_2} \mathcal{Q}_2
  \dots \abstractTransition{\mathcal R_n} \mathcal{Q}_n$ where $\mathcal{Q}_n \notmodels \neg \Bad$. The counterexample is \emph{spurious}
  if $\{\Init\} \mathcal{R}_1;\dots;\mathcal{R}_n \{\lnot \Bad\}$.
\end{definition}

Note that checking whether
$\{\Init\} \mathcal{R}_1;\dots;\mathcal{R}_n \{\lnot \Bad\}$ is
equivalent to each of the following two entailments:
\begin{eqnarray*}
  && \SP(\Init,\mathcal{R}_1;\dots;\mathcal{R}_n) = \SP(\dots \SP(\SP(\Init,
  \mathcal R_1), \mathcal R_2) \dots, \mathcal R_n) \models \neg
  \Bad \\
  && \Init\models \WP(\mathcal{R}_1;\dots;\mathcal{R}_n,\lnot\Bad) =
  \WP(\mathcal R_1, \dots
  \WP(\mathcal R_{n-1},\WP(\mathcal R_n,\lnot\Bad)) \dots )
\end{eqnarray*}

Hence we have at least two options for checking spuriousness. In both cases,
this involves intermediate predicates
$\mathcal{Q}'_1,\dots,\mathcal{Q}'_{n-1}$
(for strongest postconditions: $\mathcal Q'_1=\SP(\Init,\mathcal R_1)$ and $\mathcal{Q}'_i = \SP(\mathcal{Q}'_{i-1}, \mathcal{R}_i)$ for $1<i\leq n$;
for weakest preconditions: $\mathcal Q'_{n-1}=\WP(\mathcal R_n,\neg\Bad)$ and $\mathcal{Q}'_{i-1} = \WP(\mathcal R_i, \mathcal Q'_i)$ for $1\leq i<n$)
such that
\[ \{\Init\} \mathcal{R}_1 \{\mathcal{Q}'_1\} \mathcal{R}_2 \dots
  \mathcal{R}_{n-1} \{\mathcal{Q}'_{n-1}\} \mathcal{R}_n \{\lnot \Bad\} \]

We then augment $P$ by adding all the $\mathcal Q'_i$. In the running example $\mathcal W_1$ equals the $\WP$-based $\mathcal Q'_1$; adding it to $P$ eliminated the counterexample.
This elimination is in fact guaranteed (in a rather obious way) by the underlying theory, as formally stated by the following proposition:

\begin{propositionrep}
  \label{thm:refinement-eliminates-ce}
  Let $P$ (with $\Init,\Bad\in P$) be a set of predicates, and (considering $\Init$ as initial state) let $\mathcal{R}_1,\cdots, \mathcal{R}_n$ be a spurious counterexample to correctness w.r.t. $\Bad$. Let
  $\mathcal{Q}'_1,\dots,\mathcal{Q}'_{n-1}$ be predicates such that
  \[ \{\Init\} \mathcal{R}_1 \{\mathcal{Q}'_1\} \mathcal{R}_2 \dots
    \mathcal{R}_{n-1} \{\mathcal{Q}'_{n-1}\} \mathcal{R}_n \{\lnot \Bad\}. \]
  
  Then, in the abstract transition system based on
  $\Abs(P\cup\{\mathcal{Q}'_1,\dots,\mathcal{Q}'_{n-1}\})$ with initial state $\Init$,  the trace
    $\mathcal{R}_1,\dots,\mathcal{R}_n$ leads to a condition entailing
    $\lnot \Bad$; in other words, it is not a counterexample any more.
\end{propositionrep}


\begin{proof}
  \mbox{}
  
  \noindent\emph{Sketch:} With the new predicates, after each abstract
  step the strongest postcondition entails the intermediate predicate
  $\mathcal{Q}'_i$. Hence $\mathcal{Q}'_i$ is entailed by the
  predicate describing the current abstract state.  Hence the last
  element of the sequence will also entail $\neg \Bad$ and therefore
  it is no longer a counterexample.

  More formally: In the refined abstract transition system let
  $\mathcal{Q}_0\equiv \Init$ and assume -- by contradiction -- that,
  in the abstract transition system based on
  $\Abs(P\cup\{\mathcal{Q}'_1,\dots,\mathcal{Q}'_{n-1}\})$, there
  exists a path
  $\mathcal{Q}_0 \abstractTransition{\mathcal R_1} \mathcal{Q}_1
  \abstractTransition{\mathcal R_2} \mathcal{Q}_2 \dots
  \abstractTransition{\mathcal R_n} \mathcal{Q}_n$ such that
  $\mathcal{Q}_n$ does not entail $\lnot\Bad$.

  We define $\mathcal{Q}'_0 = \Init$, $\mathcal{Q}'_n = \lnot \Bad$
  and show that (the conjunction representing) $\mathcal{Q}_i$ entails
  $\mathcal{Q}'_i$, leading to a contradiction. Clearly
  $\mathcal{Q}_0 = \Init$ entails $\mathcal{Q}'_0 = \Init$. Now assume
  that $\mathcal{Q}_i\models \mathcal{Q}'_i$. Then
  $\SP(\mathcal{Q}_i,\mathcal{R}_{i+1}) \models
  \SP(\mathcal{Q}'_i,\mathcal{R}_{i+1}) \models \mathcal{Q}'_{i+1}$ by
  monotonicity of $\SP$ and the fact that
  $\{\mathcal{Q}'_i\} \mathcal{R}_{i+1} \{\mathcal{Q}'_{i+1}\}$. Hence
  $\mathcal{Q}_{i+1} = \alpha(\SP(\mathcal{Q}_i,\mathcal{R}_{i+1}))$
  entails $\mathcal{Q}'_{i+1}$ due to the definition of $\alpha$
  (with $\alpha(\mathcal{Q})=\overline{\mathcal{Q}}$).  \qed
\end{proof}

\subsection{Idealized algorithm}
\label{sec:idealized-algorithm}

The above brings us to the (idealized) CEGAR algorithm already discussed in \Cref{sec:motivation} and illustrated in \Cref{fig:cegar-method}. Starting
with an initial set of predicates $P = \{\Init, \Bad\}$, construct
the abstract transition system with initial state $\Init$, adding successor
states until either no new states are found or we reach a state $\mathcal{Q}_n$ that does not entail $\lnot \Bad$. In the former case, the algorithm terminates: verification succeeded, the system is correct w.r.t.\ $\Bad$. In the latter case, however, the 
sequence of rules $\mathcal{R}_1,\dots,\mathcal{R}_n$ from $\Init$ to $\mathcal{Q}_n$ is a counterexample; check whether it is spurious by computing either strongest postconditions or weakest preconditions, obtaining additional predicates $\mathcal{Q}'_1,\dots,\mathcal{Q}'_{n-1}$ as described above.

\begin{itemize}
\item If it is
  spurious, add $\mathcal{Q}_1', \dots \mathcal{Q}_n'$ to the current predicate
  set $P$ to eliminate the spurious counterexample, and restart the
  analysis.
\item If it is not spurious, the algorithm terminates:
  verification failed, the system is not correct w.r.t.\ $\Bad$.
\end{itemize}

\begin{propositionrep}
  \label{prop:correctness-idealized-algo}
  The idealized algorithm is correct in the sense that
  the system is correct if it is successful and incorrect if it
  failed. Moreover, if counterexamples are processed in ascending
  length (i.e., we always process the \emph{shortest} counterexample), it
  constitutes a semi-decision procedure.
\end{propositionrep}

\begin{proof}
  The successful output occurs when, after some number of steps, the
  algorithm generated a set of predicates such that the corresponding
  abstract transition system has only reachable states entailing
  $\neg\Bad$. By \Cref{thm:abstract-ts-correct}, this means the system
  is correct.

  A failed output results from having found a counterexample that is
  not spurious. By \Cref{def:spurious-cex}, this means that the
  following is not a valid Hoare triple:
  \[ \{\Init\} \mathcal{R}_1 \{\mathcal{Q}'_1\} \mathcal{R}_2 \dots
    \mathcal{R}_{n-1} \{\mathcal{Q}'_{n-1}\} \mathcal{R}_n \{\lnot
    \Bad\}. \] Hence
  $\SP(\Init,\mathcal{R}_1;\dots;\mathcal{R}_n)\notmodels \neg\Bad$,
  hence there exists at least one arrow satisfying $\Init$ that can be
  transformed by the rule sequence $\mathcal{R}_1;\dots;\mathcal{R}_n$
  to an arrow that does not satisfy $\neg\Bad$, hence it satisfies
  $\Bad$. This implies that the system is incorrect.

  If the system is incorrect, there exists a counterexample of length
  $m$ witnessing this. As counterexamples are processed in ascending
  length, this counterexample is eventually found and the algorithm
  terminates.
  \qed
\end{proof}

\subsection{Practical algorithm}
\label{sec:undecidability-issues}

So far we have based our definitions on entailment ($\models$),
assuming that it is somehow computable, when in fact the entailment
problem is in general undecidable. This affects both the abstraction
$\alpha \colon \Cond_0/{\equiv} \to \Abs(P)$ and the check for
spuriousness.

A practical implementation of the entailment check can only give
approximate answers to this problem and may be unable to prove or
disprove some entailments. Hence we can only rely on ``provable
entailment'' $\hatmodels$, which is a subrelation of $\models$, not formally characterized but determined by the strength of our proof tools and the available time. Based
on this and given $\mathcal{A}\in\Cond_0$, we can define
$\widehat{\mathcal{A}} \defeq \bigland \{ \mathcal{Q}' \in \Abs(P)
\mid \mathcal{A} \hatmodels \mathcal{Q}' \}$ (the strongest condition
in $\Abs(P)$ for which we can prove that it is implied by
$\mathcal A$). Since $\hatmodels$ is a subrelation of $\models$, we
obtain $\overline{\mathcal{A}} \models \widehat{\mathcal{A}}$.

In practice, this is computed by iterating over all predicates
$\mathcal{P}\in P$ and checking whether
$\mathcal{A}\hatmodels \mathcal{P}$ or
$\mathcal{A}\hatmodels \lnot \mathcal{P}$. Taking the conjunction of
all such predicates yields $\widehat{\mathcal{A}}$. Predicates for
which we obtain no result are not included in the conjunction.

Based on that we can define
$\widehat{\alpha} \colon \Cond_0/{\equiv} \to \Abs(P)$ via
$\widehat{\alpha}(\mathcal{A}) = \widehat{\mathcal{A}}$ as a
computable function. Note that $\widehat{\alpha}$ is an
over-approximation of $\alpha$ ($\alpha\models\widehat{\alpha}$).

Using such an over-approximation $\widehat{\alpha}$ also affects the
abstract transitions.  Compared to \Cref{def:ts-for-predicates},
which used the induced over-approximation
$\SP^\#(\mathcal{Q},\mathcal R) = \alpha(\SP(\gamma(\mathcal{Q}),
\mathcal R))$, we now obtain a function
$\widehat{\SP}(\mathcal{Q},\mathcal R) \defeq
\widehat{\alpha}(\SP(\gamma(\mathcal{Q}), \mathcal R))$.  This is no
longer the induced over-approximation, however, since
$\alpha \models \widehat{\alpha}$, we have
$\SP^\#(\mathcal{Q},\mathcal R) \models
\widehat{\SP}(\mathcal{Q},\mathcal R)$.  Therefore $\widehat{\SP}$ is
a safe approximation of $\SP$ and a corresponding adaptation of
\Cref{thm:abstract-ts-correct} to such abstract transition systems still holds.

Undecidability of $\models$ also affects detection of spurious
counterexamples which involves an entailment check (see
\Cref{sec:obtaining-predicates}).  A solver might be unable to
produce a proof in reasonable time, which means that the algorithm
cannot check and eliminate this specific counterexample. Either the
method proceeds with another counterexample or stops and reports the
found counterexample to the user as a (potential) error.

Hence, while the previous procedure is a semi-decision method
(cf.\ \Cref{prop:correctness-idealized-algo}), this is now lost by the
additional level of abstraction.

\section{Implementation}

The CEGAR algorithm described in this paper has been implemented in a prototypical tool\footnote{\url{https://git.uni-due.de/sflastol/cegar-prototype}}~\cite{u:abstraktionsverfeinerung}, instantiated to $\ILC(\graphf)$, i.e., graph transformation systems.
It is a command-line based tool that verifies a given graph transformation system against given predicates $\Init, \Bad$, and automatically derives new predicates from strongest postconditions.
The tool is written in Java and makes heavy use of the graph library presented in \cite{bkmns:graph-tool-library}.
For checking satisfiability of conditions, the algorithm from \cite{sksclo:coinductive-satisfiability-nested-conditions} is used and has been adapted to cospan conditions.

Previous tests of the satisfiability checker alone have indicated that for many practical examples, the intermediate conditions quickly reach impractical sizes, leading to long runtimes, which might limit the usefulness of the tool.
There are several possible optimizations to drastically cut down the
size of intermediate results, only some of which have been implemented
in the tool so far. This is an avenue for future work.

Nevertheless, there are several examples that can be successfully
verified using the CEGAR tool in its current state.  One example (\texttt{talk\_delete2.sgf}) is a
system where a rule can delete two nodes at once, $\Init$ states that
exactly three nodes exist, and $\Bad$ states that the graph is empty
(for simplicity, edges are assumed to be absent). Another example (\texttt{talk\_outedge.sgf})
starts from a non-empty graph, adds and deletes edges and their target
nodes and asks again whether the empty graph is reachable. The tool can verify these systems to be safe
after a single refinement step using strongest postconditions, almost
instantaneously.  Further successful examples of similar complexity
are discussed in \cite{u:abstraktionsverfeinerung}.

As of now, the tool can immediately spot the error in the unrefined
version of our running example (\texttt{paper\_examples.sgf}), but is unable to verify the refined
variant in reasonable time (out of memory after 130 seconds). We hypothesize that strongest
postconditions --- as currently implemented --- do not always yield the most useful predicates for
refinement and using weakest preconditions (as we manually did in the
running example) might yield better results.
%
Furthermore, optimizations can be implemented to reduce the size of the conditions both in the CEGAR loop itself and the satisfiability checker.
In manual tests, elimination of redundant subconditions has typically resulted in a reduction of more than half of all subconditions at any given step. (A similar question is studied in \cite{nkat:conditions-model-transformation}, which reports on the simplification of constraint-guaranteeing conditions.)
An additional optimization could be to pair the existing solver with
translation to first-order formula and using off-the-shelf first-order
logic or SMT solvers.




\section{Conclusion and Future Work}
\label{sec:conclusion}

Static analysis and verification techniques for graph transformation
systems have by now been studied for at least two decades. Graph
transformation provides a flexible and powerful modelling technique,
but this comes with a trade-off with respect to verification. Due to
its inherent complexity, with features such as infinite state spaces
and the ability to model dynamic topologies, graph transformation
system pose several challenges for verification. We can not give a
complete overview over the literature, but we would like to mention
that there are contributions based on efficient state space
enumeration \cite{r:groove}, invariant checking
\cite{hw:consistency-conditional-gragra,dg:k-inductive-invariants-gts},
over-approximation \cite{r:canonical-graph-shapes,kk:augur2},
well-structured transition systems
\cite{o:resilience-wst-gts,ks:framework-wsgts-journal}, logic-based
approaches
\cite{pp:hoare-calculus-graph-programs,pp:mso-incorrectness-gp2},
techniques for termination analysis
\cite{bkz:termination-gts,oe:termination-gts-weighted}, and so on. In
many cases, these techniques were adapted from other areas such as
string and term rewriting, infinite state verification or program
analysis.

A technique that has not yet received much attention in the area of
graph transformation is the CEGAR technique adapted in this paper.
CEGAR is a well-known and widely studied method for program analysis
\cite{cgjlv:abstraction-refinement,cgjlv:abstraction-refinement-journal,hjmm:abstractions-proofs,hjms:lazy-abstraction,jm:practical-predicate-refinement,m:interpolation-smc,eks:abstr-refine-craig}. The
typical setting is transition systems (Kripke structures) respectively
program verification.

To the best of our knowledge, this is the first CEGAR approach to the verification
of graph transformation systems that is based on predicate
abstraction. In \cite{kk:cegar-gts} a CEGAR framework based on
approximated unfoldings was introduced for graph transformation
systems. This however follows a completely different approach -- not
connected to predicate abstraction -- and works only for a restricted
class of systems: in particular rules are not allowed to delete nodes
and there is no integration of application conditions.

As mentioned above, the prototypical implementation still has
scalability issues and a next step would be to introduce several
optimization techniques, in particular, for improving efficiency for
the entailment problem and covering more entailments. One possibility
could be to consider sufficient conditions for entailment that are
easier to check, similar to~\cite{rc:categories-nested-conditions}.

Another line of research is to improve the automatically generated
predicates, in terms of their size and in terms of their usefulness
(i.e., whether they help to prove the correctness of the system).
This is also related to the scalability question. In CEGAR the typical
idea is to use so-called Craig interpolants
\cite{hjmm:abstractions-proofs}. Given two formulas $\phi_1,\phi_2$ in
a program logic (referring to program variables) with
$\phi_1\models\phi_2$, the Craig interpolant (of $\phi_1,\phi_2$) is a
formula $\psi$ that satisfies $\phi_1\models\psi\models\phi_2$ and
contains only the variables present in both $\phi_1,\phi_2$. The idea
is to eliminate spurious counterexamples by computing \emph{both}
weakest preconditions and strongest postconditions (where the former
entails the latter), but adding the
Craig interpolants between both. This typically leads to more compact
and better suited predicates. In the setting of conditions it
is unclear what the analogue to Craig interpolants actually is, how
they can be computed and used. We believe that this is a promising and
potentially fruitful line of research.

In this paper we concentrated mainly on applications in the area of
graph transformation systems. However, the framework of reactive systems is
more general and encompasses also other rewriting systems, such as
ground term rewriting based on Lawvere theories. In
\cite{sksclo:coinductive-satisfiability-nested-conditions} we showed how
to compute shifts in this setting, which enable us to compute
pre- and postconditions and instantiate the entire CEGAR framework. It
would be worthwhile to further investigate the applicability and scalability
of this approach.

\medskip

\noindent\emph{Disclosure of interests:} The
authors have no competing interests to declare that are relevant to
the content of this article.

\bibliographystyle{splncs04}
\bibliography{references}
\renewcommand{\appendixbibliographystyle}{plain}

\end{document}